\documentclass[aps,pra,onecolumn,notitlepage]{revtex4-1}
\usepackage{epsfig,amssymb,amsmath,mathrsfs,amsfonts,amsthm,graphicx,float,color}
\usepackage[ruled,linesnumbered,procnumbered]{algorithm2e}
\usepackage[colorlinks=true]{hyperref}
\usepackage[capitalise]{cleveref}
\usepackage[active]{srcltx}
\usepackage{subfigure}
\usepackage{comment}
\usepackage{etoolbox}
\usepackage{pbox}

 \def\map#1{\mathcal #1}
\def\<{\langle}\def\>{\rangle}

\def\Tr{\operatorname{Tr}}\def\:{\hbox{\bf
    :}}

\def\spc#1{\mathcal{#1}}
\def\rank{\mathsf{rank}}

\def\set#1{\mathsf{#1}}

\newtheorem{theo}{{Theorem}}

\crefname{theo}{Theorem}{Theorems}

\newtheorem{lem}{{Lemma}}

\crefname{lem}{Lemma}{Lemmas}

\newtheorem{prop}{{Proposition}}

\crefname{prop}{Proposition}{Propositions}

\crefname{cor}{Corollary}{Corollaries}

\newtheorem{definition}{{Definition}}

\crefname{definition}{Definition}{Definitions}

\newtheorem{assumption}{{Assumption}}

\crefname{assumption}{Assumption}{Assumptions}

\SetAlgorithmName{Algorithm}{Algorithm}{List of algorithms}
\SetProcNameSty{texttt}
\SetAlgoFuncName{Function}{Function}

\makeatletter
\AtBeginEnvironment{function}{\let\c@algocf\c@function}
\makeatother

\newcommand{\ket}[1]{\left|#1\right\rangle}
\newcommand{\kket}[1]{\left|#1\middle\rangle\!\right\rangle}
\newcommand{\skket}[1]{|#1\rangle\!\rangle}
\newcommand{\bra}[1]{\left\langle#1\right|}

\newcommand{\bbra}[1]{\left\langle\!\middle\langle#1\right|}
\newcommand{\sbbra}[1]{\langle\!\langle#1|}
\newcommand{\braket}[2]{\left\langle#1\middle|#2\right\rangle}
\newcommand{\ketbra}[1]{\ket{#1}\bra{#1}}
\newcommand{\kketbbra}[1]{\kket{#1}\bbra{#1}}
\newcommand{\skketbbra}[1]{\skket{#1}\sbbra{#1}}
\newcommand{\bbrakket}[2]{\left\langle\!\left\langle#1\middle|#2\right\rangle\!\right\rangle}

\newcommand{\ind}{\DataSty{ind}}
\newcommand{\True}{\KwSty{true}}
\newcommand{\False}{\KwSty{false}}
\newcommand{\independent}{\FuncSty{independence}}

\usepackage{tikz}
\usetikzlibrary{calc}

\usepackage{array} 
\usepackage{etoolbox}

\usetikzlibrary{decorations.pathreplacing, positioning, shapes.misc}
\tikzset{tensor/.style={rectangle,color=black,draw=black,fill=white,thick,
                    inner sep=1pt,minimum size=5mm}}
\tikzset{tensor2h/.style={rectangle,color=black,draw=black,fill=white,thick,
                    inner sep=1pt,minimum width=5mm,minimum height=10mm}}
\tikzset{parameter/.style={rectangle,color=black,draw=black,fill=black!10,thick,
                    inner sep=1pt,minimum size=5mm}}
\tikzset{virtual/.style={rectangle,inner sep=1pt,minimum size=5mm}}
\tikzset{prepare/.style={rounded rectangle, rounded rectangle east arc=none,color=black,draw=black,fill=white,thick,inner sep=1pt,minimum size=5mm}}
\tikzset{measure/.style={rounded rectangle, rounded rectangle west arc=none,color=black,draw=black,fill=white,thick,inner sep=1pt,minimum size=5mm}}
\tikzset{
    triple3/.style args={[#1] in [#2] in [#3]}{
        #1,preaction={preaction={draw,#3},draw,#2}
    }
}
\tikzset{triple/.style={triple3={[line width=0.125mm,black] in [line width=2mm,white] in [line width=2.25mm,black]}}}

\newcommand{\measurement}{\draw (55:0.5) arc (55:125:0.5); \draw (80:0.25) -- (80:0.6);}
\newcommand{\drawground}{\draw[thick] (-0.15,0) -- (0.15,0);\draw[thick] (-0.10,-0.05) -- (0.10,-0.05);\draw[thick] (-0.05,-0.1) -- (0.05,-0.1);}

\newenvironment{mathtikz}[1][]{\begin{array}{c}\begin{tikzpicture}[#1]}{\end{tikzpicture}\end{array}}
\makeatletter

\pgfkeys{/pgf/.cd,
  rectangle corner radius north west/.initial=10pt,
  rectangle corner radius north east/.initial=10pt,
  rectangle corner radius south west/.initial=10pt,
  rectangle corner radius south east/.initial=10pt
}
\newif\ifpgf@rectanglewrc@donecorner@
\def\pgf@rectanglewithroundedcorners@docorner#1#2#3#4#5{%
  \edef\pgf@marshal{%
    \noexpand\pgfintersectionofpaths
      {%
        \noexpand\pgfpathmoveto{\noexpand\pgfpoint{\the\pgf@xa}{\the\pgf@ya}}%
        \noexpand\pgfpathlineto{\noexpand\pgfpoint{\the\pgf@x}{\the\pgf@y}}%
      }%
      {%
        \noexpand\pgfpathmoveto{\noexpand\pgfpointadd
          {\noexpand\pgfpoint{\the\pgf@xc}{\the\pgf@yc}}%
          {\noexpand\pgfpoint{#1}{#2}}}%
        \noexpand\pgfpatharc{#3}{#4}{#5}%
      }%
    }%
  \pgf@process{\pgf@marshal\pgfpointintersectionsolution{1}}%
  \pgf@process{\pgftransforminvert\pgfpointtransformed{}}%
  \pgf@rectanglewrc@donecorner@true
}
\pgfdeclareshape{rectangle with rounded corners}
{
  \inheritsavedanchors[from=rectangle] 
  \inheritanchor[from=rectangle]{north}
  \inheritanchor[from=rectangle]{north west}
  \inheritanchor[from=rectangle]{north east}
  \inheritanchor[from=rectangle]{center}
  \inheritanchor[from=rectangle]{west}
  \inheritanchor[from=rectangle]{east}
  \inheritanchor[from=rectangle]{mid}
  \inheritanchor[from=rectangle]{mid west}
  \inheritanchor[from=rectangle]{mid east}
  \inheritanchor[from=rectangle]{base}
  \inheritanchor[from=rectangle]{base west}
  \inheritanchor[from=rectangle]{base east}
  \inheritanchor[from=rectangle]{south}
  \inheritanchor[from=rectangle]{south west}
  \inheritanchor[from=rectangle]{south east}

  \savedmacro\cornerradiusnw{%
    \edef\cornerradiusnw{\pgfkeysvalueof{/pgf/rectangle corner radius north west}}%
  }
  \savedmacro\cornerradiusne{%
    \edef\cornerradiusne{\pgfkeysvalueof{/pgf/rectangle corner radius north east}}%
  }
  \savedmacro\cornerradiussw{%
    \edef\cornerradiussw{\pgfkeysvalueof{/pgf/rectangle corner radius south west}}%
  }
  \savedmacro\cornerradiusse{%
    \edef\cornerradiusse{\pgfkeysvalueof{/pgf/rectangle corner radius south east}}%
  }

  \backgroundpath{%
    \northeast\advance\pgf@y-\cornerradiusne\relax
    \pgfpathmoveto{}%
    \pgfpatharc{0}{90}{\cornerradiusne}%
    \northeast\pgf@ya=\pgf@y\southwest\advance\pgf@x\cornerradiusnw\relax\pgf@y=\pgf@ya
    \pgfpathlineto{}%
    \pgfpatharc{90}{180}{\cornerradiusnw}%
    \southwest\advance\pgf@y\cornerradiussw\relax
    \pgfpathlineto{}%
    \pgfpatharc{180}{270}{\cornerradiussw}%
    \northeast\pgf@xa=\pgf@x\advance\pgf@xa-\cornerradiusse\southwest\pgf@x=\pgf@xa
    \pgfpathlineto{}%
    \pgfpatharc{270}{360}{\cornerradiusse}%
    \northeast\advance\pgf@y-\cornerradiusne\relax
    \pgfpathlineto{}%
    \pgfpathclose
  }

  \anchor{before north east}{\northeast\advance\pgf@y-\cornerradiusne}
  \anchor{after north east}{\northeast\advance\pgf@x-\cornerradiusne}
  \anchor{before north west}{\southwest\pgf@xa=\pgf@x\advance\pgf@xa\cornerradiusnw
    \northeast\pgf@x=\pgf@xa}
  \anchor{after north west}{\northeast\pgf@ya=\pgf@y\advance\pgf@ya-\cornerradiusnw
    \southwest\pgf@y=\pgf@ya}
  \anchor{before south west}{\southwest\advance\pgf@y\cornerradiussw}
  \anchor{after south west}{\southwest\advance\pgf@x\cornerradiussw}
  \anchor{before south east}{\northeast\pgf@xa=\pgf@x\advance\pgf@xa-\cornerradiusse
    \southwest\pgf@x=\pgf@xa}
  \anchor{after south east}{\southwest\pgf@ya=\pgf@y\advance\pgf@ya\cornerradiusse
    \northeast\pgf@y=\pgf@ya}

  \anchorborder{%
    \pgf@xb=\pgf@x
    \pgf@yb=\pgf@y%
    \southwest%
    \pgf@xa=\pgf@x
    \pgf@ya=\pgf@y%
    \northeast%
    \advance\pgf@x by-\pgf@xa%
    \advance\pgf@y by-\pgf@ya%
    \pgf@xc=.5\pgf@x
    \pgf@yc=.5\pgf@y%
    \advance\pgf@xa by\pgf@xc
    \advance\pgf@ya by\pgf@yc%
    \edef\pgf@marshal{%
      \noexpand\pgfpointborderrectangle
      {\noexpand\pgfqpoint{\the\pgf@xb}{\the\pgf@yb}}
      {\noexpand\pgfqpoint{\the\pgf@xc}{\the\pgf@yc}}%
    }%
    \pgf@process{\pgf@marshal}%
    \advance\pgf@x by\pgf@xa%
    \advance\pgf@y by\pgf@ya%
    \pgfextract@process\borderpoint{}%
    \pgf@rectanglewrc@donecorner@false
    \southwest\pgf@xc=\pgf@x\pgf@yc=\pgf@y
    \advance\pgf@xc\cornerradiussw\relax\advance\pgf@yc\cornerradiussw\relax
    \borderpoint
    \ifdim\pgf@x<\pgf@xc\relax\ifdim\pgf@y<\pgf@yc\relax
      \pgf@rectanglewithroundedcorners@docorner{-\cornerradiussw}{0pt}{180}{270}{\cornerradiussw}%
    \fi\fi
    \ifpgf@rectanglewrc@donecorner@\else
      \southwest\pgf@yc=\pgf@y\relax\northeast\pgf@xc=\pgf@x\relax
      \advance\pgf@xc-\cornerradiusse\relax\advance\pgf@yc\cornerradiusse\relax
      \borderpoint
      \ifdim\pgf@x>\pgf@xc\relax\ifdim\pgf@y<\pgf@yc\relax
       \pgf@rectanglewithroundedcorners@docorner{0pt}{-\cornerradiusse}{270}{360}{\cornerradiusse}%
      \fi\fi
    \fi
    \ifpgf@rectanglewrc@donecorner@\else
      \northeast\pgf@xc=\pgf@x\relax\pgf@yc=\pgf@y\relax
      \advance\pgf@xc-\cornerradiusne\relax\advance\pgf@yc-\cornerradiusne\relax
      \borderpoint
      \ifdim\pgf@x>\pgf@xc\relax\ifdim\pgf@y>\pgf@yc\relax
       \pgf@rectanglewithroundedcorners@docorner{\cornerradiusne}{0pt}{0}{90}{\cornerradiusne}%
      \fi\fi
    \fi
    \ifpgf@rectanglewrc@donecorner@\else
      \northeast\pgf@yc=\pgf@y\relax\southwest\pgf@xc=\pgf@x\relax
      \advance\pgf@xc\cornerradiusnw\relax\advance\pgf@yc-\cornerradiusnw\relax
      \borderpoint
      \ifdim\pgf@x<\pgf@xc\relax\ifdim\pgf@y>\pgf@yc\relax
       \pgf@rectanglewithroundedcorners@docorner{0pt}{\cornerradiusnw}{90}{180}{\cornerradiusnw}%
      \fi\fi
    \fi
  }
}

\makeatother

\begin{document}
\title{
Efficient Algorithms for Causal Order Discovery in Quantum Networks
  }
\author{Ge Bai}
\affiliation{QICI Quantum Information and Computation Initiative, Department of Computer Science, The University of Hong Kong, Pokfulam Road, Hong Kong}
\affiliation{HKU-Oxford Joint Laboratory for Quantum Information and Computation}
\author{Ya-Dong Wu}
\affiliation{QICI Quantum Information and Computation Initiative, Department of Computer Science, The University of Hong Kong, Pokfulam Road, Hong Kong}
\affiliation{HKU-Oxford Joint Laboratory for Quantum Information and Computation}
\author{Yan Zhu}
\affiliation{QICI Quantum Information and Computation Initiative, Department of Computer Science, The University of Hong Kong, Pokfulam Road, Hong Kong}
\affiliation{HKU-Oxford Joint Laboratory for Quantum Information and Computation}
\author{Masahito Hayashi}
\affiliation{Shenzhen Institute for Quantum Science and Engineering, Southern University of Science and Technology, Shenzhen 518055, China \looseness=-1}
\affiliation{Guangdong Provincial Key Laboratory of Quantum Science and Engineering, Southern University of Science and Technology, Shenzhen 518055, China}
\affiliation{Shenzhen Key Laboratory of Quantum Science and Engineering, Southern University of Science and Technology, Shenzhen 518055, China}
\affiliation{Graduate School of Mathematics, Nagoya University, Nagoya, 464-8602, Japan}
\author{Giulio Chiribella}
\affiliation{QICI Quantum Information and Computation Initiative, Department of Computer Science, The University of Hong Kong, Pokfulam Road, Hong Kong}
\affiliation{Department of Computer Science, University of Oxford, Parks Road, Oxford OX1 3QD, United Kingdom\looseness=-1}
\affiliation{HKU-Oxford Joint Laboratory for Quantum Information and Computation}
\affiliation{Perimeter Institute For Theoretical Physics, 31 Caroline Street North, Waterloo N2L 2Y5, Ontario, Canada\looseness=-1}
\affiliation{The University of Hong Kong Shenzhen Institute of Research and Innovation, Yuexing 2nd Rd Nanshan, Shenzhen 518057, China}

\begin{abstract}
Given black-box access to the input and output systems, we develop the first efficient quantum causal order discovery algorithm with polynomial query complexity with respect to the number of systems. We model the causal order with quantum combs, and our algorithm outputs the order of inputs and outputs that the given process is compatible with.
Our algorithm searches for the last input and the last output in the causal order, removes them, and iteratively
repeats the above procedure until we get the order of all inputs and outputs. Our method guarantees a polynomial
running time for quantum combs with a low Kraus rank, namely processes with low noise and little information loss.
For special cases where the causal order can be inferred from
local observations, we also propose algorithms that have lower query complexity and only require local state preparation and local
measurements.
Our algorithms will provide efficient ways to detect and optimize available transmission paths in quantum communication networks, as well as methods to verify quantum circuits and to discover the latent structure of multipartite quantum systems.
\end{abstract}


\maketitle

\section{Introduction}

Quantum networks \cite{kimble2008quantum,elliott2002building} are the backbone of large-scale quantum communication and computation. The study of quantum networks provides tools to analyze interactions between users, quantum channels and devices, with applications to quantum key distribution \cite{elliott2002building}, quantum distributed computation \cite{buhrman2003distributed} and quantum cloud computing \cite{barz2012demonstration}.
To support the applications, an efficient and robust quantum network is indispensable.

To ensure the robustness of quantum communication, we need to keep the communication channel stable from noise and errors, which could be solved by quantum error-correcting codes \cite{steane1996error} and protocols resistant to misaligned reference frames \cite{bartlett2003classical,aolita2007quantum}.
However, at a larger scale, other sources of instability emerge from the structure of the network.
In a quantum communication network,
data are not usually transmitted in a simple point-to-point manner: parties are connected by
repeaters and routers that serve as intermediate transmission nodes, whose availability may be
frequently changing due to network traffic and environmental noise, causing dynamical changes
in the structure of the network. To determine the optimal path to transmit data, one has to
detect those structural changes frequently, and adaptively adjust the transmission paths.

Formally, the signalling of information can be modeled as cause-effect relations. The identification of cause-effect relations, i.e. causal order discovery, is crucial for a wide range of applications in science and society. This problem has been extensively studied in the classical scenario \cite{spirtes2000causation,heinze2018causal}, while a quantum version of the causal order discovery problem has been formulated in Ref. \cite{costa2016quantum}.
Basic cases involving causal relations between a few inputs and outputs has been studied in Refs. \cite{ried2015quantum,chiribella2019quantum}.
Ref. \cite{giarmatzi2018quantum} deals with the case of many inputs and outputs. It formulates quantum causal orders as quantum combs \cite{chiribella2008quantum}, and proposes a classical algorithm based on the full classical description of the process. However, in a quantum network that is frequently changing, the classical description is not known in advance, and obtaining such a description is practically difficult, often requiring a process tomography \cite{chuang1997prescription} which can take exponential time.


In this article, we adopt the formulation of quantum causal orders as quantum combs, and propose the first efficient quantum causal order discovery algorithm for many-system quantum processes with black-box queries to the process. 
Our algorithm searches for the last input and the last output in the causal order, removes them, and iteratively repeats the above procedure until we get  the order of all inputs and outputs.
Our method guarantees a polynomial running time for quantum combs with a low Kraus rank, namely processes with low noise and little information loss.
We also propose algorithms with a lower query complexity for cases where the causal order can be inferred from local observations, for example, when each input has a non-trivial
influence on all outputs after it, and when the comb is a tensor product of single-system channels. These algorithms only require local state preparation and local measurements, and the number of uses of the process could grow logarithmically with the number of input and output systems.


\section{Preliminaries}




For a Hilbert space $\spc{H}$, let $L(\spc{H})$ be the set of linear operators on $\spc{H}$ and $S(\spc{H})$ be the set of density operators. A quantum process is described by a completely positive trace-preserving (CPTP) linear map, $\map{C}$, also known as a quantum channel. We consider quantum processes involving multiple systems, and we label the input systems as $A_1,\dots,A_n$, and output systems as $B_1,\dots,B_n$. Such a quantum process is a map $\map{C}: L(\spc{H}_{A_1}\otimes\dots\otimes\spc{H}_{A_n}) \to L(\spc{H}_{B_1}\otimes\dots\otimes\spc{H}_{B_n})$ from the tensor product of input Hilbert spaces to the tensor product of output Hilbert spaces.

To describe a quantum process with a certain causal structure, we adopt the notion of {\em quantum combs} \cite{chiribella2008quantum}:
\begin{definition}\cite{chiribella2008quantum}
A {\em quantum comb} is a concatenation of $n$ channels with memory as shown in \autoref{fig:comb}, where each channel is called a {\em tooth} of the comb. Define $\set{Comb}[(A_1,B_1),\dots,(A_n,B_n)]$ as the set of quantum combs with $n$ teeth, where the $i$-th tooth has input $A_i$ and output $B_i$  as shown in \autoref{fig:comb}.
\end{definition}
\begin{figure}[H]
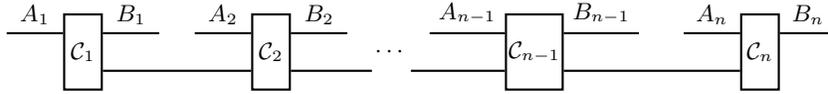

$$\begin{mathtikz}
    \def\nx{{"","1","2","n-2","n-1","n"}}
    \node[tensor2h] (U1) at (1,0) {$\map{C}_1$};
    \node[tensor2h] (U2) at (3.5,0) {$\map{C}_2$};
    \node[virtual] (U3) at (5.1,0) {$\cdots$};
    \node[tensor2h] (U4) at (7,0) {$\map{C}_{n-1}$};
    \node[tensor2h] (U5) at (10,0) {$\map{C}_n$};
    \def\lasti{1};
    \foreach \i[remember=\i as \lasti] in {2,3,4,5}
        \draw[thick] ($(U\lasti.east)-(0,0.25)$) -- ($(U\i.west)-(0,0.25)$);
    \foreach \i in {1,2,5}
    {
        \draw[thick] ($(U\i.east)+(0,0.25)$) -- node[above]{$B_{\pgfmathparse{\nx[\i]}\pgfmathresult}$} +(0.75,0);
        \draw[thick] ($(U\i.west)+(0,0.25)$) -- node[above]{$A_{\pgfmathparse{\nx[\i]}\pgfmathresult}$} +(-0.75,0);
    }
    \foreach \i in {4}
    {
        \draw[thick] ($(U\i.east)+(0,0.25)$) -- node[above]{$B_{\pgfmathparse{\nx[\i]}\pgfmathresult}$} +(1,0);
        \draw[thick] ($(U\i.west)+(0,0.25)$) -- node[above]{$A_{\pgfmathparse{\nx[\i]}\pgfmathresult}$} +(-1,0);
    }
\end{mathtikz}$$
\caption{\label{fig:comb}Quantum comb with $n$ teeth.}
\end{figure}

The following set of conditions determines whether a channel $\map C$ is a quantum comb with a given order of inputs and outputs. Let $d_{A_i}:=\dim \spc{H}_{A_i}$ and $d_{B_j}:=\dim \spc{H}_{B_j}$ be the dimensions of each input and output system. The conditions are based on the Choi state \cite{choi1975completely} of a channel, defined as $C := (\map{C}\otimes\map{I})(\kketbbra{I}/d_{A_1,...,A_n})$ where $\kket{I} := \sum_{i=1}^n \ket{i}\ket{i}  \in \bigotimes_{i=1}^n \spc{H}_{A_i} \otimes \bigotimes_{i=1}^n \spc{H}_{A_i}$ and $d_{A_1,...,A_n}:=\prod_{i=1}^n d_{A_i}$.
We write state $C$ with a subscript to denote the marginal state of $C$ on systems listed in the subscript. For example, $C_{A_1,\dots,A_n,B_1,\dots,B_k}:= \Tr_{B_{k+1},\dots,B_n}[C]$.
\begin{prop} \label{prop:comb}  \cite{chiribella2008quantum}
    Let $C$ be the Choi state of a channel $\map{C}$. $\map{C} \in \set{Comb}[(A_1,B_1),\dots,(A_n,B_n)]$ if and only if
    \begin{align}\label{eq:comb}
        C_{A_1,\dots,A_n,B_1,\dots,B_k} = C_{A_1,\dots,A_k,B_1,\dots,B_k} \otimes \frac{I_{A_{k+1},\dots,A_{n}}}{\prod_{i=k+1}^n d_{A_i}}, ~ \forall k=0,\dots,n-1
    \end{align}
    where the equation reads $C_{A_1,\dots,A_n} = I_{A_1,\dots,A_n}/\prod_{i=1}^n d_{A_i}$ for $k=0$.
\end{prop}

Given a quantum channel $\map{C}: L(\spc{H}_{A_1}\otimes\dots\otimes\spc{H}_{A_n}) \to L(\spc{H}_{B_1}\otimes\dots\otimes\spc{H}_{B_n})$, our goal is to discover its causal order, namely to identify an ordering of the inputs and outputs, $(A_{\sigma(1)},B_{\pi(1)}),\dots,(A_{\sigma(n)},B_{\pi(n)})$ with $\sigma$ and $\pi$ being permutations of $\{1,\dots,n\}$, such that
\begin{align} \label{eq:goal_exact}
    \map{C} \in \set{Comb}[(A_{\sigma(1)},B_{\pi(1)}),\dots,(A_{\sigma(n)},B_{\pi(n)})].
\end{align}
To ensure that the problem always has an answer, we assume that $\map{C}$ is guaranteed to be a quantum comb with some ordering, namely  $\map{C} \in \set{Comb}[(A_{\sigma'(1)},B_{\pi'(1)}),\dots,(A_{\sigma'(n)},B_{\pi'(n)})]$, where $\sigma'$ and $\pi'$ are unknown permutations so that $\sigma=\sigma', \pi=\pi'$ is always an answer for Eq. (\ref{eq:goal_exact}).

To allow for noise and errors that are inevitable in any quantum operation, we formalize an approximate version of Eq. (\ref{eq:goal_exact}), that we only require $\map{C}$ to approximately have the causal order $(A_{\sigma(1)},B_{\pi(1)}),\dots,(A_{\sigma(n)},B_{\pi(n)})$. Let $\varepsilon$ be an error threshold, we relax our goal as finding permutations $\sigma$ and $\pi$ such that
\begin{align} \label{eq:goal_approx}
    \exists \map{D} \in \set{Comb}[(A_{\sigma(1)},B_{\pi(1)}),\dots,(A_{\sigma(n)},B_{\pi(n)})], \|\map{C} - \map{D}\| \leq \varepsilon
\end{align}
where $\|\cdot\|$ is some distance measure to be determined later. To quantize the efficiency of our algorithms, we will focus on the query complexity, namely the number of black-box accesses to the channel $\map{C}$, and the additional running time, including quantum and classical computation time measured by the number of elementary quantum gates and classical operations.





It is convenient to convert quantum data to classical data so that classical data analysis techniques could be adopted. When doing this conversion, we need to ensure that the information encoded in the quantum data are preserved, for which purpose we introduce the {\em informationally complete} POVM \cite{prugovevcki1977information}.
\begin{definition} \label{def:IC}
    A POVM $\{P_\alpha\}$ on Hilbert space $\spc{H}$ is {\em informationally complete} if its elements form a spanning set of linear operators on $\spc{H}$, namely for any $X\in L(\spc{H})$, there exist complex numbers $\{p_\alpha\}$ such that $X=\sum_\alpha p_\alpha P_\alpha$.

    We say a set of states $\{\psi_\alpha\}$ is {\em informationally complete} if there exists an informationally complete POVM $\{P_\alpha\}$ such that ${\psi_\alpha} = P_\alpha/\Tr[P_\alpha], \forall \alpha$.
\end{definition}
For efficiency considerations, we prefer to choose an informationally complete POVM with the minimal number of elements. To form a spanning set of $L(\spc{H})$, the minimal informationally complete POVM contains $\dim(\spc{H})^2$ elements, which are linearly independent. 
A typical example of such POVMs is the symmetric informationally complete POVM (SIC-POVM) \cite{renes2004symmetric}, which has been found for most relevant low-dimensional Hilbert spaces.

The frame operators defined below will be used to characterize the sensitivity of measurement outcomes about the state being measured.
\begin{definition} \label{def:frame}
    The {\em frame operator} of an informationally complete POVM $\{P_\alpha\}$ on Hilbert space $\spc{H}$ is defined as $F :=\sum_\alpha\kketbbra{P_\alpha}\in L(\spc{H}\otimes\spc{H})$, where $\kket{P_\alpha} := (P_\alpha \otimes I) \kket{I} = \sum_{i,j}(P_\alpha)_{ij}\ket{i}\ket{j} \in \spc{H}\otimes\spc{H}$, where $\kket{I}:=\sum_i\ket{i}\ket{i}$ is the unnormalized maximally entangled state.

    For an informationally complete set of states $\{\psi_\alpha\}$, the frame operator is defined as $F':=\sum_\alpha\kketbbra{\psi_\alpha}$ with $\kket{\psi_\alpha}:= ({\psi_\alpha} \otimes I) \kket{I}$. 
\end{definition}

A frame operator with a larger minimum eigenvalue indicates that the outcome probabilities are more sensitive to any change of the state being measured, and the POVM is more efficient for obtaining the full information of the state. We will see the consequences of the minimum eigenvalue of frame operators in the analysis of our algorithms.

In the following sections, we will use the following distance measures for operators $\rho,\sigma \in L(\spc{H})$:
\begin{enumerate}
    \item Trace distance, defined as $\frac12\|\rho - \sigma\|_1$, where $\|X\|_1 := \Tr\left[\sqrt{X^\dag X}\right]$ denotes the Schatten 1-norm, also known as the trace norm and the nuclear norm. For quantum states in any dimension, the trace distance between them is no larger than 1. We will more often use the trace distance without the $1/2$ factor.
    \item Hilbert-Schmidt distance, defined as $\|\rho - \sigma\|_2$, where $\|X\|_2:= \sqrt{\Tr[X^\dag X]}$ denotes the Schatten 2-norm, also known as the Frobenius norm.
\end{enumerate}

The Hilbert-Schmidt distance is dimension-dependent, and is related to the trace distance by the following \cite{coles2019strong}
\begin{align}\label{eq:norm12}
    \|\rho - \sigma\|_1^2 \leq \rank(\rho-\sigma) \|\rho - \sigma\|_2^2 \leq (\rank(\rho)+\rank(\sigma)) \|\rho - \sigma\|_2^2 \,.
\end{align}

To measure the distance between channels, we will use the diamond norm, also known as the completely bounded trace norm. For two channels $\map{C},\map{D}: L(\spc{H}_A) \to L(\spc{H}_B)$, their diamond norm distance is defined as \cite{kitaev2002classical}
\begin{align} \label{eq:diamond_norm}
    \| \map{C} - \map{D} \|_\diamond := \max_{\rho \in S(\spc{H}_A \otimes \spc{H}_A)} \| (\map{C}\otimes \map{I}_A) (\rho) - (\map{D}\otimes \map{I}_A) (\rho) \|_1
\end{align}
where $\map{I}_A:L(\spc{H}_A) \to L(\spc{H}_A)$ is the identity map.

\section{Efficient causal order discovery algorithm for general causal order}

To determine the causal order, one may first make a hypothesis on the causal order and perform tests that accept or reject the hypothesis.
To decide whether $\map{C}$ is a quantum comb of given order, we need to test the equalities in the form of Eq. (\ref{eq:comb}). To conveniently present such equalities, we define the notion of independence for quantum states as the following:
\begin{definition}
    For a state $\rho$, we say two disjoint sets of subsystems $S_{\rm A}$ and $S_{\rm B}$ are {\em independent} if $\rho_{S_{\rm A}\cup S_{\rm B}} = \rho_{S_{\rm A}} \otimes \rho_{S_{\rm B}}$, where $\rho_S$ is the marginal state of $\rho$ on the subsystems in $S$.
\end{definition}

The condition Eq. (\ref{eq:comb}) is then equivalent to that for the Choi state $C$, the sets of systems $\{A_1,\dots,A_k,B_1,\dots,B_k\}$ and $\{A_{k+1},\dots,A_{n}\}$ are independent for every $k$.


To test the conditions in the form of Eq. (\ref{eq:comb}), we use the {\em SWAP test} \cite{buhrman2001quantum}, a quantum circuit that estimates $\Tr[\rho\sigma]$ for two given quantum states $\rho$ and $\sigma$. The circuit is depicted in \autoref{fig:SWAP}.

\newcommand{\cS}{\mathsf{CSWAP}}
\newcommand{\swap}{\mathsf{SWAP}}

\begin{figure}[H]
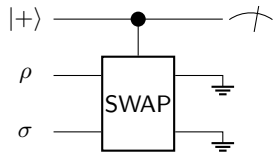

$$\begin{array}{ccc}
    \begin{mathtikz}
        \node[tensor, minimum height=12.5mm] (swap) at (0,-0.375) {$\swap$};
        \node (phi1) at (-1.5,0) {$\rho$};
        \node (phi2) at (-1.5,-0.75) {$\sigma$};
        \node (c) at (-1.5,0.75) {$\ket{+}$};
        \node (right) at (1.125, 0) {};
        \node (left) at (-1.125, 0) {};
        \node (measure) at (1.5,0.75) {\tikz{\measurement}};
        \draw (phi1-|left) -- (phi1-|swap.west);
        \draw (phi2-|left) -- (phi2-|swap.west);
        \coordinate (grd1) at ($(right|-phi1) - (0,0.15)$);
        \coordinate (grd2) at ($(right|-phi2) - (0,0.15)$);
        \begin{scope}[shift={(grd1)}]\drawground\end{scope}
        \begin{scope}[shift={(grd2)}]\drawground\end{scope}
        \draw (swap.east|-phi1) -| (grd1);
        \draw (swap.east|-phi2) -| (grd2);
        \draw (c-|left) -- (c-|right);
        \fill (swap|-c) circle [radius=0.1];
        \draw (swap|-c) -- (swap);
    \end{mathtikz}
\end{array}$$
\caption{\label{fig:SWAP}The SWAP test. The circuit consists of a controlled-SWAP gate with control qubit initialized to $\ket+$. Measuring the control system under the $\{\ket{+},\ket{-}\}$ basis yields the outcome $\ket{+}$ with probability $(1+\Tr[\rho\sigma])/2$. The ground symbol means discarding the system.}
\end{figure}

The algorithm that yields an estimate of $\Tr[\rho\sigma]$ is as follows, which produces an estimate with error no more than $\varepsilon$ with probability $1-\kappa$ as shown in \autoref{lem:SWAP}.

\begin{function}[H]
    \caption{SWAPTEST($\rho$, $\sigma$, $\varepsilon$, $\kappa$)}\label{func:SWAPTEST}
    \Indentp{0.5em}
    \SetInd{0.5em}{1em}
    \SetKwInOut{Preprocessing}{Preprocessing}
    \SetKwInOut{Input}{Input}
    \SetKwInOut{Output}{Output}
    \SetKw{Next}{next}
    \SetKwData{ind}{ind}
    \SetKwData{last}{last}
    \SetKwData{purity}{purity}
    \SetKwData{maxpurity}{maxpurity}

    \ResetInOut{Input}
    \Input{Quantum states $\rho$ and $\sigma$ (accessed by oracles that generate the states), error threshold $\varepsilon$, confidence $\kappa$}
    \ResetInOut{Output}
    \Output{Approximate value of $\Tr[\rho\sigma]$}
    \BlankLine
    $N \gets \lceil 2\varepsilon^{-2}\log(2/\kappa) \rceil$\;
    Run the circuit in \autoref{fig:SWAP} for $N$ times. Let $c_+$ be the number of outcome $\ket +$\;
    Return $2c_+/N - 1$\;
\end{function}

\begin{lem}\label{lem:SWAP}
With probability $1-\kappa$, the SWAP test estimates $\Tr[\rho\sigma]$ within error $\varepsilon$.
\end{lem}
\begin{proof}
For each run, the probability that the outcome is $\ket +$ is $(1+\Tr[\rho\sigma])/2$.
By Hoeffding's inequality,
\begin{align}
\nonumber    \Pr \left[ \left| \frac{c_+}{N} - \frac{1+\Tr[\rho\sigma]}{2} \right| \leq \varepsilon/2 \right] & \geq 1-2e^{-\varepsilon^2 N /2} \\
    \Pr \left[ \left| 2c_+/N - 1 - \Tr[\rho\sigma] \right| \leq \varepsilon \right] & \geq 1-2e^{-\varepsilon^2 ( 2\varepsilon^{-2}\log(2/\kappa)) /2} = 1-\kappa
\end{align}
\end{proof}

The SWAP test is efficient in the sense that its circuit complexity scales only logarithmically with the dimension $d$ of $\rho$ or $\sigma$. Assuming the states are represented with qubits, the controlled-SWAP gate acting on $\rho$ and $\sigma$ equals to $\lceil \log d \rceil$ controlled-SWAP gates acting on each pair of corresponding qubits of $\rho$ and $\sigma$, and thus can be implemented with $O(\log d)$ gates.

Now, we can use the SWAP test to construct an algorithm that tests the independence between quantum systems. Running SWAP tests on $\rho \otimes \sigma$, $\rho \otimes \rho$, and $\sigma\otimes\sigma$, we are able to obtain estimates for $\Tr[\rho\sigma]$, $\Tr[\rho^2]$ and $\Tr[\sigma^2]$, respectively. From these estimates we can compute the Hilbert-Schmidt distance between the states as $\|\rho - \sigma\|_2^2 = \Tr[\rho^2] + \Tr[\sigma^2] - 2\Tr[\rho\sigma]$. 
Estimating the distance between $\rho$ and $\sigma$ allows us to decide whether these two states are different.


To decide whether two systems are independent, we can test whether different observations of one system would steer the state of the other system to a different marginal state, and the difference between marginal states can be tested with SWAP tests.
Specifically, we are going to perform independence tests between the input and output of a quantum channel, or equivalently, whether the channel is a constant channel. The idea is to input different states ${\psi_1},\dots,{\psi_m}$ to the channel $\map{C}$, and check whether the output states $\map{C}({\psi_1}),\dots,\map{C}({\psi_m})$ are all the same. It turns out that this approach works when the chosen input states $\{\psi_1,\dots,\psi_m\}$ form an informationally complete set, as we will show in Appendix \ref{app:unitary} \autoref{lem:ic_diamond}.


According to Eq. (\ref{eq:norm12}), to convert the Hilbert-Schmidt to trace distance which is more meaningful in state discrimination \cite{costa2016quantum}, we need to ensure that the ranks of states we are testing are not too high.
The low-rank condition can be satisfied if the comb is composed of unitary gates, and the memory size is small. Therefore, we make the following assumption:

\begin{assumption} \label{ass:unitary}
The quantum comb has a memory system of dimension $d_M$ initialized to a pure state, and every input interacts unitarily with the memory system. Formally, the quantum comb $\map{C}\in\set{Comb}[(A_{\sigma'(1)},B_{\pi'(1)}),\dots,(A_{\sigma'(n)},B_{\pi'(n)})]$ has the following form:
$$\begin{mathtikz}
    \def\nx{{"","1","2","n-2","n-1","n"}}
    \node[prepare] (psi) at (-0.5,-0.25) {$\psi_0$};
    \node[tensor2h] (U1) at (1,0) {$U_1$};
    \node[tensor2h] (U2) at (4,0) {$U_2$};
    \node[virtual] (U3) at (6,0) {$\cdots$};
    \node[tensor2h] (U4) at (8.5,0) {$U_{n-1}$};
    \node[tensor2h] (U5) at (12,0) {$U_n$};
    \draw[thick] ($(U1.west)-(0,0.25)$) -- node[below]{$M_0$} (psi);
    \def\lasti{1};
    \foreach \i[remember=\i as \lasti] in {2,3,4,5}
        \draw[thick] ($(U\lasti.east)-(0,0.25)$) -- node[below]{$M_{\pgfmathparse{\nx[\lasti]}\pgfmathresult}$} ($(U\i.west)-(0,0.25)$);
    \foreach \i in {1,2,5}
    {
        \draw[thick] ($(U\i.east)+(0,0.25)$) -- node[above]{$B_{\pi'(\pgfmathparse{\nx[\i]}\pgfmathresult)}$} +(1,0);
        \draw[thick] ($(U\i.west)+(0,0.25)$) -- node[above]{$A_{\sigma'(\pgfmathparse{\nx[\i]}\pgfmathresult)}$} +(-1,0);
    }
    \foreach \i in {4}
    {
        \draw[thick] ($(U\i.east)+(0,0.25)$) -- node[above]{$B_{\pi'(\pgfmathparse{\nx[\i]}\pgfmathresult)}$} +(1.4,0);
        \draw[thick] ($(U\i.west)+(0,0.25)$) -- node[above]{$A_{\sigma'(\pgfmathparse{\nx[\i]}\pgfmathresult)}$} +(-1.4,0);
    }
    \coordinate (grd) at ($(U5.east)-(0,0.25) + (0.25,-0.5)$);
    \draw[thick] ($(U5.east)-(0,0.25)$) -| (grd);
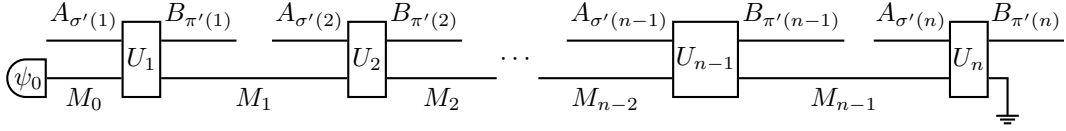
    \begin{scope}[shift={(grd)}]
        \drawground
    \end{scope}
\end{mathtikz}$$
where $\psi_0$ is a pure state, each of $U_1,U_2,\dots,U_n$ is a unitary transformation, each input or output wire has the same dimension $d_A=d_{A_1}=\dots=d_{A_n}=d_{B_1}=\dots=d_{B_n}$, and each memory system $M_i$ has the same dimension $d_M$.
\end{assumption}

\autoref{ass:unitary} does not limit the generality of the quantum process we investigate, since any quantum process with a causal order is equivalent to a unitary circuit in the form of \autoref{ass:unitary} \cite{barrett2019quantum}. The purpose of this assumption is to clarify the size of each unitary gate, which will determine the efficiency of the algorithm we are going to present.

Now we are ready to present our algorithm for combs satisfying \autoref{ass:unitary}. The algorithm runs in a recursive manner: for an $n$-tooth comb, it finds the last tooth of the comb, traces the tooth out (feeds the input with maximally mixed state and discards its output), and reduces the problem to finding the last tooth of an $(n-1)$-tooth comb. \autoref{ass:unitary} ensures that after tracing out the last tooth, the remaining comb still has a low rank and satisfies \autoref{ass:unitary}. We repeat the above procedure until we have only one tooth left, and thus obtain the order of all inputs and outputs.

To find the last tooth, we go back to to Proposition \ref{prop:comb}. If $(A_i,B_j)$ is the input and output of last tooth, denoting $A_{\neq i}$ ($B_{\neq j}$) the set of input (output) wires other than $A_i$ ($B_j$), we have
\begin{align}
\map{C} \in \set{Comb}[(A_{\neq i},B_{\neq j}),(A_i,B_j)] \iff C_{A_1,\dots,A_n,B_{\neq j}} = C_{A_{\neq i},B_{\neq j}} \otimes \frac{I_{A_i}}{d_{A_i}} \,.
\end{align}
The left hand side means that $\map{C}$ is compatible with a causal order with $(A_i,B_j)$ being the last tooth, and the right hand side means that $A_i$ is independent with all the other input and output wires except $B_j$. We enumerate all possible $i$ and $j$, and check whether $(A_i,B_j)$ can be the last tooth with an independence test. If the test is passed, we know that the channel $\map{C}$ can be written in a comb form where $(A_i,B_j)$ is the last tooth, although $\map{C}$ may intrinsically have a different ``last tooth'' since there may be multiple causal orders with which $\map{C}$ is compatible.

The algorithm to find the last tooth is as follows. To test the independence between $A_i$ and $A_{\neq i},B_{\neq j}$, we convert all input wires in $A_{\neq i}$ to output wires by inserting maximally entangled states, as shown in \autoref{fig:rhob}, and test whether the resulting channel is a constant channel.

\begin{function}[H]
    \caption{findlast($\map{O}$, $n$, $\delta$, $\kappa$)}\label{func:findlast}
    \Indentp{0.5em}
    \SetInd{0.5em}{1em}
    \SetKwInOut{Preprocessing}{Preprocessing}
    \SetKwInOut{Input}{Input}
    \SetKwInOut{Output}{Output}
    \SetKw{Next}{next}
    \SetKw{Break}{break}
    \SetKwData{ind}{ind}
    \SetKwData{last}{last}
    \SetKwData{purity}{purity}
    \SetKwData{maxpurity}{maxpurity}
    \SetKwFunction{SWAPTEST}{SWAPTEST}

    \ResetInOut{Input}
    \Input{A quantum channel oracle $\map{O}$ with input wires $A_1,\dots,A_n$ and output wires $B_1,\dots,B_n$, error threshold $\delta$, confidence $\kappa$}
    \ResetInOut{Output}
    \Output{The labels of the input and output of the last tooth of $\map{O}$}
    \BlankLine

    $\maxpurity \gets -\infty$\;
    $\varepsilon \gets \delta/4$\;
    \For{$i \gets 1$ \KwTo $n$}
    {
        \For{$j \gets 1$ \KwTo $n$}
        {\label{line:forj}
            Pick an informationally complete set of states for $A_i$, and call them $\{\psi_b\}_{b=1}^{d_{A_i}^2}$ \;
            For each $b$, prepare the circuit that generate the state $\rho_b := \Tr_{B_j} \left[ \left( \map{O}\otimes\map{I_{A'_{\neq i}}} \right) \left( \bigotimes_{k \neq i} \frac{\kketbbra{I}_{A_kA'_k}}{d_{A_k}} \otimes (\psi_b)_{A_i} \right) \right]$ as shown in \autoref{fig:rhob}\;
            $\ind \gets $ \True \;
            $p_{1} \gets $\SWAPTEST{$\rho_1$, $\rho_1$, $\varepsilon$, $\kappa$} \;
            \For{$k\gets 2$ \KwTo $d_{A_i}^2$}
            {
                $p_{k} \gets $\SWAPTEST{$\rho_k$, $\rho_k$, $\varepsilon$, $\kappa$} \;
                $p_{1k} \gets $\SWAPTEST{$\rho_1$, $\rho_k$, $\varepsilon$, $\kappa$} \;
                \If{$p_{1} + p_{k} - 2 p_{1k} > \delta$}
                {
                \tcp{The estimated Hilbert-Schmidt distance between $\rho_1$ and $\rho_k$ is larger than $\delta$}
                    $\ind \gets$ \False \;
                    \Break\tcp*{Exit the loop for $k$}
                }
            }
            \If{$\ind$}
            {
                Output $(i,j)$ and exit this function\;
            }
        }
    }
    Output ``$\map{O}$ is not a quantum comb.'' and halt\;
\end{function}

\begin{figure}
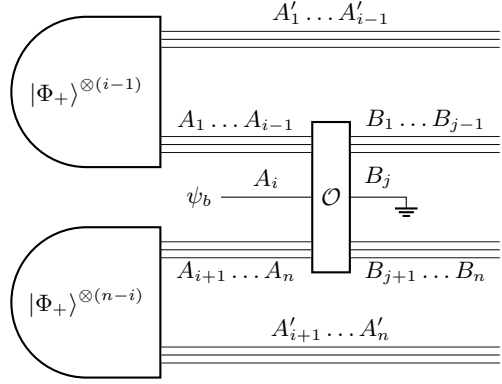

$$\begin{mathtikz}
    \node[tensor, minimum height=2cm] (O) at (0,0) {$\map{O}$};
    \node (psi) at (-1.75,0) {$\psi_b$};
    \draw (psi) -- node[above] {$A_i$} (O);
    \draw[triple] ($(O.west)+(0,0.7)$) -- node[above,yshift=0.5mm] {$A_1\dots A_{i-1}$} +(-2,0);
    \draw[triple] ($(O.west)-(0,0.7)$) -- node[below,yshift=-0.5mm] {$A_{i+1}\dots A_n$} +(-2,0);
    \node[prepare, minimum height=2cm] (E1) at (-3.25,1.4) {$\ket{\Phi_+}^{\otimes(i-1)}$};
    \node[prepare, minimum height=2cm] (E2) at (-3.25,-1.4) {$\ket{\Phi_+}^{\otimes(n-i)}$};
    \coordinate (right) at (2.25,0);
    \draw[triple] ($(O.east)+(0,0.7)$) -- node[above,yshift=0.5mm] {$B_1\dots B_{j-1}$} ($(O.east -| right)+(0,0.7)$);
    \draw[triple] ($(O.east)-(0,0.7)$) -- node[below,yshift=-0.5mm] {$B_{j+1}\dots B_n$} ($(O.east -| right)-(0,0.7)$);
    \draw[triple] ($(E1.east)+(0,0.7)$) -- node[above,yshift=0.5mm] {$A'_1\dots A'_{i-1}$} ($(E1.east -| right)+(0,0.7)$);
    \draw[triple] ($(E2.east)-(0,0.7)$) -- node[above,yshift=0.5mm] {$A'_{i+1}\dots A'_n$} ($(E2.east -| right)-(0,0.7)$);
    \draw[white] ($(right)-(0,2.25)$) -- ($(right)+(0,2.25)$);  
    \coordinate (grd) at (1,-0.15);
    \draw (O) -- node[above] {$B_j$} (O-|grd) -- (grd);
    \begin{scope}[shift={(grd)}]
        \drawground
    \end{scope}
\end{mathtikz}$$
\caption{{\bf Circuit for $\rho_b$.} $\ket{\Phi_+}:= \kket{I}/\sqrt{d_A}$ denotes the maximally entangled state.\label{fig:rhob}}
\end{figure}

Given the subroutine to find the last tooth, the main recursive algorithm is as follows.

\begin{algorithm}[H]
    \Indentp{0.5em}
    \SetInd{0.5em}{1em}
    \SetKwInOut{Preprocessing}{Preprocessing}
    \SetKwInOut{Input}{Input}
    \SetKwInOut{Output}{Output}
    \SetKwFor{Whenever}{whenever}{do}{end}
    \SetKw{KwDownto}{downto}
    \SetKwData{ind}{ind}
    \SetKwFunction{findlast}{findlast}

    \ResetInOut{Input}
    \Input{Black-box access of a quantum comb $\map{C}$ with input wires $A_1,\dots,A_n$ and output wires $B_1,\dots,B_n$, error threshold $\delta$, confidence $\kappa$}
    \ResetInOut{Output}
    \Output{Ordering of inputs and outputs} 
    \BlankLine

    Initialize $S$ to an empty list\;
    Define $\map{O} := \map{C}$ \;
    \For{$k \gets n$ \KwDownto $1$}
    {
        $(A_x,B_y) \gets$ \findlast{$\map{O}$, $k$, $\delta$, $\kappa$}\;
        Insert $(A_x,B_y)$ before the beginning of $S$\;
        Update the definition of $\map{O}$ to $\rho \mapsto \Tr_{B_y}\left[\map{O}\left(\rho \otimes \frac{I_{A_x}}{d_{A_x}}\right)\right]$ \tcp*{Trace out the last tooth}
    }
    Output $S$\;
    \caption{Causal order discovery algorithm based on \autoref{ass:unitary}}\label{alg:unitary}
\end{algorithm}

The efficiency and correctness of \autoref{alg:unitary} is given in the following theorem, whose proof is in Appendix \ref{app:unitary}:

\begin{theo} \label{thm:unitary}
Under \autoref{ass:unitary}, with probability $1-\kappa_0$, \autoref{alg:unitary} outputs a causal order $(A_{\sigma(1)},B_{\pi(1)}),\dots,(A_{\sigma(n)},B_{\pi(n)})$ such that
    \begin{align}\exists \map{D} \in \set{Comb}[(A_{\sigma(1)},B_{\pi(1)}),\dots,(A_{\sigma(n)},B_{\pi(n)})], ~ \|C - D\|_1 \leq \varepsilon_0 \label{eq:thm_unitary}\end{align}
with number of queries to $\map{C}$ in the order of
    \begin{align}
        T_{\rm query}=O\left(n^{11} d_A^{12} d_M^2 \varepsilon_0^{-8} \lambda_{\min}^{-2} \log (n d_A \kappa_0^{-1})\right)
    \end{align}
and extra processing time in the order of
    \begin{align}O(T_{\rm query}n\log d_A)\end{align}
where $\lambda_{\min}$ is the minimum eigenvalue of the frame operator of informationally complete set of states $\{\psi_k\}$ for a $d_A$-dimensional system. If we choose $\{\psi_k\}$ proportional to a SIC-POVM, we have $\lambda_{\min}^{-1} = (d_A+1)/d_A=O(1)$  \cite{renes2004symmetric}.
\end{theo}

\autoref{alg:unitary} is the first efficient algorithm for quantum causal order discovery. Compared with the algorithm in Ref. \cite{giarmatzi2018quantum}, we do not require a classical description of the process and access the quantum process in a black-box manner. Additionally, the efficiency of our algorithm is guaranteed by \autoref{thm:unitary}.

In Eq. (\ref{eq:thm_unitary}), we have used the trace norm between Choi state as a measure of distance between channels. This measures the optimal probability one can distinguish the Choi state $C$ and $D$. For the distance between channels, we often use the diamond norm distance, which measures the distinguishability between channels $\map{C}$ and $\map{D}$, where all possible entangled states on all input systems and an additional ancillary system are allowed for probing the difference between the channels. Generally $\|C - D\|_1 \leq \|\map{C} - \map{D}\|_\diamond$, meaning that a small trace distance between Choi states does not guarantee the indistinguishability of the channels. However, the trace distance is still meaningful if we limit the observer to access one pair of input and output at a time. This is because for the marginal channels on input $A_i$ and output $B_j$, denoted as $\map{C}_{A_i\to B_j}$ and  $\map{D}_{A_i\to B_j}$, one has
\begin{align}
    \| \map{C}_{A_i\to B_j} - \map{D}_{A_i\to B_j} \|_\diamond \leq d_{A_i} \|C_{A_iB_j} - D_{A_iB_j} \|_1 \leq  d_{A_i} \|C - D\|_1 \leq d_{A_i}\varepsilon_0
\end{align}
which is a small value since $d_{A_i}$ is constant. This shows that, an observer that has access only to the wires $A_i$ and $B_j$ is unable to distinguish the channels $\map{C}$ and $\map{D}$ with a high probability.



\section{Causal order discovery algorithms based on local observations}

\autoref{alg:unitary} is already efficient in the sense that its complexity scales polynomially with the number of systems. Under some assumptions on the causal order, we are able to give algorithms that have much lower query complexity and are more experimentally friendly. They require only local state preparations and local measurements.

\subsection{Independence test based on classical statistics}\label{ss:indep}

Recall that the condition Eq. (\ref{eq:comb}) for quantum combs is essentially testing the independence between systems of the Choi state. Now we consider the cases where the causal order can be inferred from local observations on the Choi state, each of which involves a pair of input and output systems.
Using informationally complete POVMs, we are able to convert quantum states to classical probability distributions without losing information. Therefore, we can perform the independence test between two quantum systems $A$ and $B$ by observing the probability distribution obtained by measuring the systems with informationally complete POVMs. Indeed, the quantum systems are independent if and only if the measurement outcomes on the systems are independent. 


To quantitatively measure the amount of correlation between two systems, we define an independence measure based on  trace distance.
\begin{definition}
For a bipartite state $\rho_{AB} \in S(\spc{H}_A \otimes \spc{H}_B)$, define
\begin{align}
    \chi_1(\rho_{AB}) := \| \rho_{AB} - \rho_A \otimes \rho_B \|_1 \,,
\end{align}
where $\rho_A:=\Tr_B[\rho_{AB}]$ and $\rho_B:=\Tr_A[\rho_{AB}]$ are the marginal states of $\rho_{AB}$.
\end{definition}
Clearly, $\chi_1(\rho_{AB}) = 0$ if and only if $\rho_{AB}$ is a product state. Our first independence test algorithm will be based on the estimation of $\chi_1(\rho_{AB})$. The idea is essentially quantum state tomography: measure $\rho_{AB}$ with informationally complete POVMs $\{P_\alpha\}$ on system $A$ and $\{Q_\beta\}$ on system $B$, use the outcome statistics to reconstruct $\rho_{AB}$, and estimate $\chi_1(\rho_{AB})$ based on the reconstructed state. The algorithm is given in \autoref{alg:c1}.

\begin{function}[H]
    \caption{estimatechi1(\pbox{5cm}{$\{(a^{(k)},b^{(k)})\}_{k=1}^{N},\{P_\alpha\},\{Q_\beta\}$})\label{alg:c1}}
    \Indentp{0.5em}
    \SetInd{0.5em}{1em}
    \SetKwInOut{Input}{Input}
    \SetKwInOut{Output}{Output}
    \SetKw{Break}{break}
    \SetKwData{matched}{matched}

    \ResetInOut{Input}
    \Input{$N$ measurement outcomes $\{(a^{(k)},b^{(k)})\}_{k=1}^{N}$ of POVM $\{P_\alpha \otimes Q_\beta\}$ on state $\rho_{AB}$}
    \ResetInOut{Output}
    \Output{An estimate of $\chi_1(\rho_{AB})$} 
    \BlankLine

    Initialize $\hat{p}_{\alpha\beta}$ to zero for all $\alpha=1,\dots,d_A^2,$ and $\beta=1,\dots,d_B^2$\;
    \For{$k \gets 1$ \KwTo $N$}
    {
            $\hat{p}_{a^{(k)}b^{(k)}} \gets \hat{p}_{a^{(k)}b^{(k)}} + 1/ N$\;
    }

    $F \gets \sum_\alpha \kketbbra{P_\alpha}$\;
    $G \gets \sum_\beta \kketbbra{Q_\beta}$\;
    $S \gets \sum_\alpha \ket{\alpha}\bbra{P_\alpha}$, where $\{\ket\alpha\}$ is a set of orthonormal vectors indexed by $\alpha$\;
    $R \gets \sum_\beta \ket{\beta}\bbra{Q_\beta}$, where $\{\ket\beta\}$ is a set of orthonormal vectors indexed by $\beta$\;
    $\ket{\hat p} \gets \sum_{\alpha,\beta} \hat p_{\alpha\beta}\ket\alpha \ket\beta$\;
    Compute operator $\hat\rho_{AB}$ by $\kket{\hat{\rho}_{AB}} = ( F^{-1} S^\dag \otimes G^{-1} R^\dag)\ket{\hat p}$ \;

    Output $\chi_1(\hat\rho_{AB})$\;
\end{function}

Now we analyze the accuracy and complexity of \autoref{alg:c1}.
Define the constant $\xi:=
\frac{ \sqrt{\lambda_{\min}(F)\lambda_{\min}(G)}}
{ \sqrt{d_A^2 d_B^2 + 4d_B^2 + 4d_A^2 } d_Ad_B}$
where $\lambda_{\min}(X)$ denotes the minimum eigenvalue of operator $X$
and $F$ and $G$ are defined in \autoref{alg:c1}. Then we define
\begin{align}\label{eq:kappaepsilon}
    \kappa(\varepsilon)
    := 2 (d_A^2d_B^2+d_A^2+d_B^2) e^{-2 \xi^2 \varepsilon^2 N} .
\end{align}
The error of the estimate $\chi_1(\hat\rho_{AB})$ is given in the following lemma, whose proof is in Appendix \ref{app:alg_ind}.
\begin{lem}\label{lem:alg_ind}
    For any positive real number $\varepsilon>0$, over the $N$ measurement outcomes on independent copies of $\rho_{AB}$, with probability $1-\kappa_0$, the output of \autoref{alg:c1} estimates $\chi_1(\rho_{AB})$ with the following error bound
    \begin{align}\label{eq:c1_bound}
        |\chi_1(\rho_{AB}) - \chi_1(\hat\rho_{AB})| \leq \varepsilon \,.
    \end{align}
    That is, the inequality
    \begin{align}\label{eq:c1_bound-ie}
    \rm{Pr} (   |\chi_1(\rho_{AB}) - \chi_1(\hat\rho_{AB})| > \varepsilon )
    \le \kappa(\varepsilon)
    \end{align}
    holds for any positive real number $\varepsilon>0$.
    Here, $\rm{Pr} ( C)$ expresses the probability that the condition $C$ is satisfied.
\end{lem}

In other words, to reach confidence $1-\kappa_0$ and error $\varepsilon_0$, the number of samples $N$ is of order
\begin{align} \label{eq:alg_ind_N}
    N = O\left(d_A^4 d_B^4 \lambda_{\min}^{-1}(F) \lambda_{\min}^{-1}(G) \varepsilon_0^{-2} \log(d_A d_B \kappa_0^{-1}) \right) \,.
\end{align}
If one choose $\{P_\alpha\}$ and $\{Q_\beta\}$ to be SIC-POVMs, $\lambda_{\min}^{-1}(F) = d_A(d_A+1) = O(d_A^2)$ and $\lambda_{\min}^{-1}(G) = d_B(d_B+1) = O(d_B^2)$ \cite{renes2004symmetric}. In this case, the number of samples $N$ is polynomial in $d_A$ and $d_B$.

\autoref{alg:c1} could be used on quantum channels to discover the correlation between input and output wires. To characterize a quantum channel, we could transform it to its Choi state by feeding the input with half of a maximally entangled state, and apply \autoref{alg:c1} on the Choi state. To save the cost of creating and maintaining entangled states, we could use \autoref{alg:c1} in an alternative way. Let $\{\psi_\alpha\}_{\alpha=1}^{d_A^2}$ be an informationally complete set of states such that $\psi_\alpha = P^T_\alpha/\Tr[P_\alpha]$. One samples $\psi_\alpha$ from $\{\psi_\alpha\}_{\alpha=1}^{d_A^2}$ with probability $\Tr[P_\alpha]/d_A$, apply quantum channel $\map{C}: L(\spc{H}_A)\to L(\spc{H}_B)$ on $\psi_\alpha$, and measure the output with $\{Q_\beta\}_{\beta=1}^{d_B^2}$. This process generates the same statistics for $(\alpha,\beta)$ according to the following equation:
\begin{align}\label{eq:channel_statistics}
    \frac{1}{d_A}\Tr[P_\alpha]\Tr[Q_\beta \map{C}(\psi_\alpha)] = \frac{1}{d_A}\Tr[Q_\beta \map{C}(P^T_\alpha)] = \Tr\left[\frac{1}{d_A}\sum_{i,j} P_\alpha\ket{i}\bra{j} \otimes Q_\beta \map{C}(\ket{i}\bra{j})\right] = \Tr[(P_\alpha \otimes Q_\beta)C]
\end{align}
where $C := \frac{1}{d_A} \sum_{i,j} \ket{i}\bra{j} \otimes \map{C}(\ket{i}\bra{j})$ is the Choi state of channel $\map{C}$. The left hand side of Eq. (\ref{eq:channel_statistics}) is the probability of sampling $\psi_\alpha$ times the probability of outcome $\beta$ on the output conditioned on the input being $\psi_\alpha$, and the right hand side of Eq. (\ref{eq:channel_statistics}) is the probability of outcome $(\alpha,\beta)$ if one measure the Choi state $C$ with POVM $\{P_\alpha \otimes Q_\beta\}$.

\subsection{Inferring causal order form independence of input-output pairs}

Now we adopt \autoref{alg:c1} to causal order discovery. By testing the correlation between every pair of input and output wires, one is able to infer the causal orders: if output $B_j$ is correlated to input $A_i$, $B_j$ must be after $A_i$. In this section, we study cases where such reasoning yields the full causal order.

We first introduce an algorithm that discovers the correlations between every pair of input and output wires. The idea is to first collect all the required classical data, and deduce all the correlations according to the statistics. During this procedure, the data are reused for different input-output pairs, leading to fewer queries to the channel $\map{C}$.



\begin{function}[H]
    \caption{independence($\map{C}, N, \chi_-$)\label{alg:order_linear2}\label{alg:ind}}
    \Indentp{0.5em}
    \SetInd{0.5em}{1em}
    \SetKwInOut{Input}{Input}
    \SetKwInOut{Output}{Output}
    \SetKwFunction{estimatechi}{estimatechi1}

    \ResetInOut{Input}
    \Input{Black-box access of a quantum comb $\map{C}$ with input wires $A_1,\dots,A_n$ and output wires $B_1,\dots,B_n$}
    \ResetInOut{Output}
    \Output{An array $\ind$ where $\ind_{i,j}$ being \True{} indicates $A_i$ and $B_j$ are independent}
    \BlankLine

    Pick informationally complete linearly independent POVMs $\{P_{\alpha}^{i}\}_{\alpha=1}^{d_{A_i}^2}$ for each input $A_i$ and $\{Q_{\beta}^{j}\}_{\beta=1}^{d_{B_j}^2}$ for each output $B_j$\; 
    Define ${\psi_\alpha^{i}}:= (P_{\alpha}^{i})^T/\Tr[P_{\alpha}^{i}],~ \forall i=1,\dots,n,~\alpha=1,\dots,d_{A_i}^2$\;
    Generate $N$ random $n$-tuples $\{(a_1^{(k)},a_2^{(k)},\dots,a_n^{(k)})\}_{k=1}^{N}$ where $a_i^{(k)}$ is chosen from $\{1,\dots,d_{A_i}^2\}$ with probability $\Pr[a_i^{(k)}=\alpha] = \Tr[P_{\alpha}^{i}]/d_{A_i}$\;
    \For{$k \gets 1$ \KwTo $N$}
    {
        Apply $\map{C}$ to input state $\psi^1_{a_1^{(k)}} \otimes \psi^2_{a_2^{(k)}} \dots \otimes \psi^n_{a_n^{(k)}}$ \;
        For each $j$, measure $j$-th output system with $\{Q_{\beta}^{j}\}_{\beta=1}^{d_{B_j}^2}$. Let the outcomes be $(b_1^{(k)},b_2^{(k)},\dots,b_n^{(k)}), ~ b_j^{(k)} \in \{1,\dots,d_{B_j}^2\}$ \;
    }
    \For{$i \gets 1$ \KwTo $n$}
    {
        \For{$j \gets 1$ \KwTo $n$}
        {
            \uIf{$\estimatechi(\{(a_i^{(k)},b_j^{(k)})\}_{k=1}^{N}, \{ P_\alpha^i \}, \{ Q_\beta^j \}) \leq \chi_-$}
            {
                $\ind_{i,j} \gets \True$ \;
            }
            \Else
            {
                $\ind_{i,j} \gets \False$ \;
            }
        }
    }
    Output $\ind$\;
\end{function}

To understand \autoref{alg:ind}, we observe that $\{(a_i^{(k)},b_j^{(k)})\}_{k=1}^{N}$ are sampled from the same distribution as the outcomes of measuring $C_{A_i,B_j}$ with $\{P_\alpha^i \otimes Q_\beta^j\}$, and $\FuncSty{estimatechi1}$ produces an estimate of $\chi_1(C_{A_i,B_j})$. Conditioned on all $n^2$ number of calls to $\FuncSty{estimatechi1}$ succeed, which has probability no less than $1-n^2\kappa_0$, all the estimations $\chi_1(C_{A_i,B_j})$ have error no greater than $\varepsilon_0$. We summarize these observations in the following lemma:

\begin{lem}\label{lem:ind}
With probability $1-n^2\kappa$, \autoref{alg:ind} produces an output satisfying:
\begin{enumerate}
    \item if $\ind_{i,j} = \False$, then $\chi_1(C_{A_i,B_j}) > \chi_- -\varepsilon$
    \item if $\ind_{i,j} = \True$, then $\chi_1(C_{A_i,B_j}) \leq \chi_- + \varepsilon$
\end{enumerate}
where $\kappa$ and $\varepsilon$ are related in the form of Eq. (\ref{eq:kappaepsilon}) with $d_A := \max_i d_{A_i}$ and $d_B := \max_j d_{B_j}$.
\end{lem}

Now we study how \autoref{alg:ind} could help with causal order discovery.
With some assumptions, the data collected by \autoref{alg:ind} is sufficient to give the exact causal order. The first case is given by the following assumption:

\begin{assumption} \label{ass:totalorder}
For the quantum comb $\map{C}\in\set{Comb}[(A_{\sigma'(1)},B_{\pi'(1)}),\dots,(A_{\sigma'(n)},B_{\pi'(n)})]$, there exists a constant $\chi_{\min}$ such that each input-output pair of systems $A_{\sigma'(i)}$ and $B_{\pi'(j)}$ satisfies the following:
\begin{enumerate}
    \item if $j<i$, $A_{\sigma'(i)}$ and $B_{\pi'(j)}$ are independent, and $\chi_1(C_{A_{\sigma'(i)},B_{\pi'(j)}} ) = 0$;
    \item if $j\geq i$, $A_{\sigma'(i)}$ and $B_{\pi'(j)}$ are correlated, and $\chi_1(C_{A_{\sigma'(i)},B_{\pi'(j)}} ) \geq \chi_{\min}$.
\end{enumerate}
\end{assumption}
\autoref{ass:totalorder} indicates a non-trivial correlation between any input-output pair that the input is before the output. In other words, for any $j \geq i$, $C_{A_{\sigma'(i)},B_{\pi'(j)}} \neq C_{A_{\sigma'(i)}} \otimes C_{B_{\pi'(j)}}$. Meanwhile, \autoref{ass:totalorder} defines a total order of the input (output) wires, and ensures a unique causal order that $\map{C}$ is compatible with, namely the only correct answer for Eq. (\ref{eq:goal_exact}) is $\sigma=\sigma'$ and $\pi=\pi'$. Under this assumption, we observe that the input $A_{\sigma'(k)}$ is correlated with $n-k+1$ output wires including every output $B_{\pi(i)}$ with $i\geq k$, and is independent with the other outputs. By counting the number of output wires that an input is correlated with, we can determine which tooth the input belongs to, and sort the input wires. The same method can be used to order the output wires. Following this logic, we can adopt the procedure in \autoref{alg:order_linear}.

\begin{algorithm}[H]
    \caption{Quantum causal order discovery algorithm based on \autoref{ass:totalorder}}\label{alg:order_linear}
    \Indentp{0.5em}
    \SetInd{0.5em}{1em}
    \SetKwInOut{Preprocessing}{Preprocessing}
    \SetKwInOut{Input}{Input}
    \SetKwInOut{Output}{Output}
    \SetKwFor{Whenever}{whenever}{do}{end}

    \ResetInOut{Input}
    \Input{Black-box access of a quantum comb $\map{C}$ with input wires $A_1,\dots,A_n$ and output wires $B_1,\dots,B_n$, number of queries $N$}
    \ResetInOut{Output}
    \Output{Ordering of inputs and outputs} 
    \BlankLine

    $\ind \gets \independent(\map{C}, N, \chi_{\min}/2)$\;
    \For{$k \gets 1$ \KwTo $n$}
    {
        $c_A(k) \gets | \{ i | \ind_{k,i}=\False \}|$ \;
        $c_B(k) \gets | \{ i | \ind_{i,k}=\False \}|$ \;
    }
    Sort $\{A_1,\dots,A_n\}$ in descending order of $c_A(k)$ as $\{A_{\sigma(1)},\dots,A_{\sigma(n)}\}$\;
    Sort $\{B_1,\dots,B_n\}$ in ascending order of $c_B(k)$ as $\{B_{\pi(1)},\dots,B_{\pi(n)}\}$\;
    Output $(A_{\sigma(1)},B_{\pi(1)}),\dots,(A_{\sigma(n)},B_{\pi(n)})$\;
\end{algorithm}


Here the threshold $\chi_-$ is set to $\chi_-=\chi_{\min}/2$. Setting $\varepsilon_0 = \chi_{\min}/3$, by \autoref{lem:ind}, with probability $1-n^2 \kappa_0$, the independence tests produce the ground truth: (i) if $\ind_{i,j}=\True$, then $\chi_1(C_{A_i,B_j}) \leq \chi_- + \varepsilon_0 < \chi_{\min}$, and according to \autoref{ass:totalorder}, $A_i$ and $B_j$ are independent; (ii) if $\ind_{i,j}=\False$, then $\chi_1(C_{A_i,B_j}) > \chi_- - \varepsilon_0 > 0$, and according to \autoref{ass:totalorder}, $A_i$ and $B_j$ are correlated.
In this case, \autoref{ass:totalorder} further ensures that $c_A(1),\dots,c_A(n)$ ($c_B(1),\dots,c_B(n)$) are distinct and \autoref{alg:order_linear} always gives the correct causal order.

The running time of \autoref{alg:order_linear} mostly attributes to \autoref{alg:ind}. \autoref{alg:ind} contains two main parts: the query part and the estimation part. The query part invokes the channel $N$ times, and each time $n$ state preparation and $n$ measurements are performed. Here we assume that state preparation on a $d$-dimensional Hilbert space takes time $O(d)$ \cite{plesch2011quantum} and performing a POVM with $m$ outcomes takes time $O(m^2)$, according to the Naimark's dilation theorem that converts $m$-outcome POVM to projective measurement via a $m$-dimensional unitary gate \cite{peres2006quantum} implementable with $O(m^2)$ basic gates \cite{mottonen2004quantum}.  To perform the POVM with $m=d_B^2$ elements, we need $O(d_B^4)$ gates, and thus the query part takes a total time $O(Nn(d_A+d_B^4))$. The estimation part invokes \autoref{alg:c1} for $n^2$ times, and each call to \autoref{alg:c1} takes time $O(N+d_A^4d_B^4)$, $O(N)$ for computing $\hat{p}$ and $O(d_A^4d_B^4)$ for handling matrices and vectors. We don't count the time of inverting the matrix $F$ ($G$) since they are known before the execution of the algorithm and their inverses can be computed in advance. To sum up, the total running time of \autoref{alg:order_linear} is $O(Nn(n+d_A+d_B^4)+n^2d_A^4d_B^4)$, where the first term dominates if we take $N=\Omega(d_A^4d_B^4)$.

Therefore, we have the following theorem:

\begin{theo}\label{thm:order_linear}
For a quantum comb $\map{C} \in \set{Comb}[(A_{\sigma(1)},B_{\pi(1)}),\dots,(A_{\sigma(n)},B_{\pi(n)})]$ satisfying \autoref{ass:totalorder}, with probability $1-\kappa$, \autoref{alg:order_linear} outputs the correct causal order $(A_{\sigma(1)},B_{\pi(1)}),\dots,(A_{\sigma(n)},B_{\pi(n)})$ with number of queries to $\map{C}$ in the order of
\begin{align} \label{eq:order_linear}
    N = O\left(d_A^4 d_B^4 \lambda_{\min}^{-2} \chi_{\min}^{-2} \log(n d_A d_B \kappa^{-1}) \right)
\end{align}
and running time $O( Nn(n + d_A + d_B^4))$, where
\begin{align}
d_A := \max_i d_{A_i}, ~d_B := \max_j d_{B_j}, ~\lambda_{\min}:=\min\left\{ \min_i \lambda_{\min}\left(\sum_\alpha\skketbbra{P_\alpha^i} \right),\min_j \lambda_{\min}\left(\sum_\beta\skketbbra{Q_\beta^j}\right)\right\}
\end{align}

\end{theo}

\autoref{alg:ind} could also be adopted to the case where the comb has no memory. The assumption is the following:

\begin{assumption} \label{ass:memoryless}
The quantum comb $\map{C}\in\set{Comb}[(A_{\sigma'(1)},B_{\pi'(1)}),\dots,(A_{\sigma'(n)},B_{\pi'(n)})]$ has no memory. In other words, it is the tensor product of $n$ channels, $\map{C} = \bigotimes_{i=1}^n \map{C}_i$ with $\map{C}_i: \spc{H}_{A_{i }}\to \spc{H}_{B_{\pi''(i)}}$ for some permutation $\pi''$.
\end{assumption}

Since each output is related to at most one input and each input affects at most one output, after the independence tests, we can obtain the causal order by matching the input-output pairs as in the following algorithm.
\begin{algorithm}[H]
    \caption{Quantum causal order discovery algorithm based on \autoref{ass:memoryless}\label{alg:memoryless}}
    \Indentp{0.5em}
    \SetInd{0.5em}{1em}
    \SetKwInOut{Input}{Input}
    \SetKwInOut{Output}{Output}
    \SetKw{Break}{break}
    \SetKwData{True}{true}
    \SetKwData{False}{false}
    \SetKwData{ind}{ind}
    \SetKwData{matched}{matched}
    \SetKwFunction{independent}{independence}

    \ResetInOut{Input}
    \Input{Black-box access of a quantum comb $\map{C}$ with input wires $A_1,\dots,A_n$ and output wires $B_1,\dots,B_n$, number of queries $N$, threshold $\chi_-$}
    \ResetInOut{Output}
    \Output{Ordering of inputs and outputs} 
    \BlankLine

    $\ind \gets \independent(\map{C}, N, \chi_-)$\;
    Initialize Boolean arrays $\matched_A$ and $\matched_B$ to $\False$\;
    \For{$i \gets 1$ \KwTo $n$}
    {
        \For{$j \gets 1$ \KwTo $n$}
        {
            \If{$\ind_{i,j} = \False$}
            {
                $\pi(i) \gets j$\;
                $\matched_A(i) \gets \True$\;
                $\matched_B(j) \gets \True$\;
                \Break\tcp*{Exit the loop for $j$}
            }
        }
    }\label{line:memoryless_match}
    \For{$i \gets 1$ \KwTo $n$}
    {
        \If{$\matched_A(i) = \False$}
        {
            \For{$j \gets 1$ \KwTo $n$}
            {
                \If{$\matched_B(j) = \False$}
                {
                    $\pi(i) \gets j$\;
                    $\matched_A(i) \gets \True$\;
                    $\matched_B(j) \gets \True$\;
                    \Break\tcp*{Exit the loop for $j$}
                }
            }
        }
    }
    Output $(A_{1},B_{\pi(1)}),\dots,(A_{n},B_{\pi(n)})$\;
\end{algorithm}
\autoref{alg:memoryless} first matches the pairs $(A_i,B_j)$ that are correlated, and according to \autoref{ass:memoryless}, $A_i$ and $B_j$ must belong to the same tooth. For the inputs and outputs left unmatched by the first step, they can be matched arbitrarily, since the inputs and outputs have no correlation and are compatible with any causal order.

If we further assume that for each correlated input-output pair $(A_i,B_j)$, one has $\chi_1(\rho_{A_i,B_j}) \geq \chi_{\min}$, then \autoref{alg:memoryless} produces an exact answer (\ref{eq:goal_exact}) if one take $\chi_- = \chi_{\min}/3$, and the query complexity is given by Eq. (\ref{eq:order_linear}). On the other hand, if such assumption of $\chi_{\min}$ does not exist, the algorithm still produces an answer, which is approximate (\ref{eq:goal_approx}). The error and running time of \autoref{alg:memoryless} is given in the following theorem, whose proof is in Appendix \ref{app:memoryless}.

\begin{theo}\label{thm:memoryless}
For a quantum comb $\map{C} \in \set{Comb}[(A_{\sigma'(1)},B_{\pi'(1)}),\dots,(A_{\sigma'(n)},B_{\pi'(n)})]$ satisfying \autoref{ass:memoryless}, with probability $1-\kappa$, \autoref{alg:memoryless} outputs a causal order $(A_1,B_{\pi(1)}),\dots,(A_n,B_{\pi(n)})$ satisfying
\begin{align}
    \exists \map{D} \in \set{Comb}[ (A_1,B_{\pi(1)}),\dots,(A_n,B_{\pi(n)}) ], ~ \|\map{C} - \map{D} \|_\diamond \leq \varepsilon
\end{align}
with number of queries to $\map{C}$ in the order of
\begin{align}
    N = O\left(n^2 d_A^6 d_B^4 \lambda_{\min}^{-2} \varepsilon^{-2} \log(n d_A d_B \kappa^{-1}) \right)
\end{align}
and running time $O( Nn(n + d_A + d_B^4))$, where
\begin{align}
d_A := \max_i d_{A_i}, ~d_B := \max_j d_{B_j}, ~\lambda_{\min}:=\min\left\{ \min_i \lambda_{\min}\left(\sum_\alpha\skketbbra{P_\alpha^i} \right),\min_j \lambda_{\min}\left(\sum_\beta\skketbbra{Q_\beta^j}\right)\right\}
\end{align}

\end{theo}

\subsection{Limitation of local observations}


The independence test \autoref{alg:order_linear2} in the previous section is only able to test the independence between pairs of input and output wires. This approach does not solve the general case, since there are cases where testing between only two wires would never produce a correct causal order, such as the comb in \autoref{fig:notsimple}. 

\begin{figure}[H]
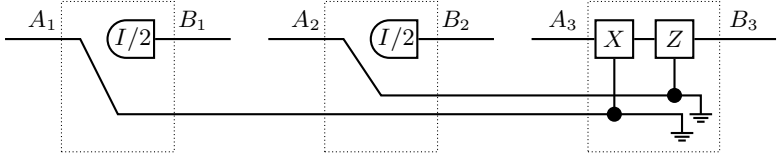

$$\begin{mathtikz}
    \coordinate (in1) at (-1,1);
    \coordinate (out1) at (2,1);
    \coordinate (in2) at (2.5,1);
    \coordinate (out2) at (5.5,1);
    \coordinate (in3) at (6,1);
    \coordinate (out3) at (9.25,1);
    \coordinate (grd1) at (8,-0.25);
    \coordinate (grd2) at (8.25,0);
    \node[prepare] (p1) at ($(out1)-(1.3,0)$) {$I/2$};
    \node[prepare] (p2) at ($(out2)-(1.3,0)$) {$I/2$};
    \draw[thick] (p1) -- node[above] {$B_1$} (out1);
    \draw[thick] (p2) -- node[above] {$B_2$} (out2);
    \draw[thick] (in1) -- node[above] {$A_1$} ++(1,0) -- ++(0.5,-1) -- ++(1,0) -| node (line1) {} (grd1);
    \draw[thick] (in2) -- node[above] {$A_2$} ++(1,0) -- ++(0.5,-0.75) -- ++(1,0) -| node (line2) {} (grd2);
    \begin{scope}[shift={(grd1)}]\drawground\end{scope}
    \begin{scope}[shift={(grd2)}]\drawground\end{scope}
    \node[tensor] (X) at ($(in3)+(1.1,0)$) {$X$};
    \node[tensor] (Z) at ($(X)+(0.8,0)$) {$Z$};
    \draw[thick] (in3) -- node[above] {$A_3$} (X) -- (Z) -- +(0.5,0) -- node[above] {$B_3$} (out3);
    \draw[thick] (X) -- (X|-line1);
    \fill (X|-line1) circle [radius=0.1];
    \draw[thick] (Z) -- (Z|-line2);
    \fill (Z|-line2) circle [radius=0.1];
    \draw[densely dotted] ($(in1)+(0.75,0.5)$) rectangle ($(out1|-line1)+(-0.75,-0.5)$);
    \draw[densely dotted] ($(in2)+(0.75,0.5)$) rectangle ($(out2|-line1)+(-0.75,-0.5)$);
    \draw[densely dotted] ($(in3)+(0.75,0.5)$) rectangle ($(out3|-line1)+(-0.75,-0.5)$);
\end{mathtikz}$$
\caption{\label{fig:notsimple} The channel acts on three qubits and outputs three qubits. $I$ is the identity operator and $I/2$ is the maximally mixed state. $X$ and $Z$ denote Pauli matrices, and here they are controlled gates with control systems $A_1$ and $A_2$, respectively. The dotted boxes indicates the teeth of the comb. The ground symbol means discarding the system. For any input $A_i$ and output $B_j$, the marginal Choi state $C_{A_iB_j}=(I\otimes I)/2$ is a product state, and thus observations on input-output pairs give no information about the causal order. In turn, $C_{A_1A_2A_3B_3} \neq C_{A_1A_2B_3}\otimes I_{A_3}/2$ and is not a product state, indicating $\map{C}$ is not compatible to arbitrary causal orders. The causal order can be inferred if one observes more than two wires simultaneously.}
\end{figure}

\section{Discussion}


In this article, we give the first efficient quantum causal order discovery algorithm whose complexity scales polynomially with the number of systems. We model the causal order with quantum combs, and formulate the problem as looking for an ordering of teeth that a given channel is compatible with. Our algorithm access the quantum process in a black-box way, and runs in polynomial time whenever the process has a low Kraus rank, corresponding to processes with low noise and little information loss.
We also propose a second approach that involves only local state preparation and local measurements, and infers the causal order by collecting statistics of the channel's input-output correlations. This approach features a low query complexity, but is limited to cases where the causal order can be inferred from local observations, for example, when each input has a non-trivial influence on all outputs after it, and when the comb is a tensor product of single-system channels.

The algorithms mentioned above, however, do not cover all possible quantum combs. This is not surprising, since it is also difficult to infer general causal orders in the classical scenario. In classical causal order discovery, we often formulate the problem by graphical models and solve it by structure learning algorithms such as PC algorithm \cite{spirtes2000causation}, Greedy Equivalence Search \cite{chickering2002optimal}, and Max-Min Hill Climbing \cite{tsamardinos2006max}. However, these algorithms do not cover all possible instances because they require assumptions of the input data such as causal sufficiency, meaning that no variables are hidden, and causal faithfulness, meaning that the conditional independence between any two variable sets must be visible from the d-separation condition in causal graph. 
In the worst case, the complexity of these algorithms becomes exponential.

The cases we have not solved involve quantum combs that are composed of channels with higher Kraus rank and  and at the same time have a causal order that cannot be inferred from local observations. In this case, our first algorithm may be inefficient since it requires the comb to be low-rank, and the approach based on local observations is unable to discover the causal order. Yet from another point of view, the comb having a higher rank means that there are more information loss and random noise in the process, and thus the correlations are weakened by the loss and overwhelmed by the noise, making it difficult for any observer to infer the causal order. With this idea, we can assign an operational meaning to the ``correct'' causal order by limiting the power of the observer: channel $\map{C}$ is compatible with causal order $(A_{\sigma(1)},B_{\pi(1)}),\dots,(A_{\sigma(n)},B_{\pi(n)})$ if any observer, who has limited knowledge about $\map{C}$ and has limited computing power, cannot distinguish $\map{C}$ from another channel $\map{D}\in\set{Comb}[(A_{\sigma(1)},B_{\pi(1)}),\dots,(A_{\sigma(n)},B_{\pi(n)})]$. This is in contrast to the objective in Eq.~(\ref{eq:goal_approx}), where we have inserted the diamond norm distance (\autoref{thm:memoryless}) or the trace distance (\autoref{thm:unitary}), which operationally measures the distinguishability of two channels for an observer who knows the classical description of the channels.
The idea of limiting the observer is similar to probably approximately correct (PAC) learning \cite{aaronson2007learnability,cheng2015learnability,caro2020pseudo}, where the observer can only make polynomial number of observations sampled from a predetermined distribution. The application of PAC learning to causal order discovery requires a new framework of causal orders outside the scope of this article, and is left as a future work.

Efficient causal order discovery algorithms are crucial to ensure the stability of quantum networks. With such algorithms, we are able to detect and fix the instability emerging from the structure changes of the network. Combined with techniques to stabilize point-to-point communication such as error-correction codes, the quantum network will serve as a reliable middle layer connecting the applications
and the physical network, giving solutions to errors and interfaces for quantum protocols and algorithms \cite{kimble2008quantum,elliott2002building,buhrman2003distributed,barz2012demonstration}.

The applications of causal order discovery are not limited to quantum networks. Our independence tests will also
be applicable to quantum devices and circuits for testing whether they have correct links and input-output
correlations. This will be an important quality assurance procedure for the production of reliable components
in quantum computers and networks. Causal order discovery can be also used to detect the latent structure of quantum systems, by applying the causal order discovery algorithms to detect the correlations in multipartite states.
Knowing the causal structure of the system may allow us to ignore unnecessary correlations and represent states efficiently with, for example, tensor networks \cite{fannes1992finitely,verstraete2008matrix}, which allow efficient tomography \cite{cramer2010efficient}, simulation \cite{verstraete2008matrix,shi2006classical,vidal2008class} and compression \cite{bai2020quantum} of multipartite states.

\medskip
{\bf Acknowledgements.} 
This work was supported by the National Natural Science Foundation of China through grant 11675136, the Hong Kong Research Grant Council through grant 17300918, and though the Senior Research Fellowship Scheme SRFS2021-7S02, the Croucher Foundation, and the John Templeton Foundation through grant 61466, The Quantum Information Structure of Spacetime (qiss.fr). Research at the Perimeter Institute is supported by the Government of Canada through the Department of Innovation, Science and Economic Development Canada and by the Province of Ontario through the Ministry of Research, Innovation and Science. The opinions expressed in this publication are those of the authors and do not necessarily reflect the views of the John Templeton Foundation.
The work of M. Hayashi was supported in part by Guangdong Provincial Key Laboratory (Grant No. 2019B121203002).

\bibliography{causal}

\appendix

\section{Proof of \autoref{thm:unitary}} \label{app:unitary}

The coefficients $\{p_x\}$ in \autoref{def:IC} can be computed with the following lemma:
\begin{lem} \label{lem:IC}
For an informationally complete POVM $\{P_x\}$, let $F :=\sum_x\kketbbra{P_x}$ be its frame operator. Then any linear operator $X\in L(\spc{H})$ can be decomposed as $X=\sum_x p_x P_x$ where $p_x = \bbra{P_x} F^{-1} \kket{X}$. The same is true for an informationally complete set of states $\{\ket{\psi_x}\}$: $X=\sum_x p'_x {\psi_x}$ where $p'_x = \bbra{\psi_x} F'^{-1} \kket{X}$, $F':=\sum_x\kketbbra{\psi_x}$ and $\kket{\psi_x}:= ({\psi_x} \otimes I) \kket{I}$.
\end{lem}
\begin{proof}
    \begin{align}
    X = \sum_x p_x P_x
    \iff& (X \otimes I) \kket{I}= \sum_x  p_{x} (P_x \otimes I) \kket{I} \\
     \iff& \kket{X} = \sum_x  p_{x} \kket{P_x},
    \end{align}
    where $\{ \kket{P_x} \}$ forms a basis of space $\map{H} \otimes \map{H}$, and $p_{x}$ are the coefficients of $\kket{X}$ in this basis.
    From
    \begin{align}
    \kket{X} = & F F^{-1} \kket{X}\\
    =&\sum_x \kketbbra{P_x} F^{-1} \kket{X} \\
    =&\sum_x \bbra{P_x} F^{-1} \kket{X} \kket{P_x},
    \end{align}
    we obtain the expression of $p_{x}$. The expression of $p'_x$ is identical by substituting $P_x$ with ${\psi_x}$.
\end{proof}

Since $\{P_x\}$ forms a basis for linear operators in $\spc{H}$, $\{\kket{P_x}\}$ forms a basis for $\spc{H}\otimes\spc{H}$, and $F$ is always full-rank.

\begin{lem} \label{lem:ic_diamond}
Let $\map{C}: S(\map{H}_{A}) \to S(\map{H}_{B})$ be a channel . Let $\{{\psi_k}\}_{k=1}^{d_{A}^2}$ be an informationally complete set of states for system $A$.
Let $\rho_k := \map{C}( {\psi_k} )$. If $\forall k, ~ \|\rho_k - \rho_1\|_1 \leq \varepsilon$, then there exists a constant map $\map D: S(\map{H}_{A}) \to S(\map{H}_{B})$, $\map{D}(\rho) := \rho_1 \Tr[\rho]$ such that
\begin{align}
    \|\map C - \map D\|_\diamond \leq d_{A}^2 \lambda_{\min}^{-1/2}  \varepsilon
\end{align}
where $D$ is the Choi state of $\map{D}$, and $\lambda_{\min}$ is the minimum eigenvalue of the frame operator of $\{{\psi_k}\}_{k=1}^{d_{A}^2}$.
\end{lem}

\begin{proof}
    The diamond norm can be evaluated as:
    \begin{align}
        \|\map C - \map D\|_\diamond = \sup_{\ket{\Phi}} \| ((\map C - \map D) \otimes \map I_{A'})(\ketbra{\Phi}) \|_1
    \end{align}
    where $\ket{\Phi}$ is a pure state on systems $A$ and $A'$, and $A'$ is a reference system with $\dim\map{H}_{A'} = \dim\map{H}_{A} = d_{A}$. Consider the Schmidt decomposition of $\ket{\Phi}$ as
    \begin{align}
        \ket{\Phi} = \sum_{k=1}^r s_k \ket{\alpha_k}_A \ket{\beta_k}_{A'}
    \end{align}
    where $r\leq d_{A}$, $\{\ket{\alpha_k}_A\}$ and $\{\ket{\beta_k}_{A'}\}$ are sets of orthonormal states, $\sum_k s_k^2 = 1$, $s_k>0, \forall k$. Since $\{\psi_k\}_{k=1}^{d_{A}^2}$ is an informationally complete set, by \autoref{lem:IC}, one has
    \begin{align}
        \ket{\alpha_i}\bra{\alpha_j} = \sum_{k=1}^{d_{A}^2}  p_{ijk} {\psi_k} \,,
    \end{align}
    where $p_{ijk} = \bbra{\psi_k} F^{-1} \ket{\alpha_i}\ket{\overline{\alpha_j}}$ with $F :=\sum_k\kketbbra{\psi_k}$ being the frame operator.
    Then
    \begin{align}
        \ketbra{\Phi} = &~ \sum_{i,j=1}^r s_i s_j \ket{\alpha_i}\bra{\alpha_j} \otimes \ket{\beta_i}\bra{\beta_j} \\
        = &~ \sum_{i,j=1}^r \sum_{k=1}^{d_{A}^2} s_i s_j p_{ijk} {\psi_k} \otimes \ket{\beta_i}\bra{\beta_j} \,.
    \end{align}
    Note that $\forall i,j$,
    \begin{align}
        \sum_{k=1}^{d_{A}^2} |p_{ijk}| \leq &~ \sqrt{d_{A}^2 \sum_{k=1}^{d_{A}^2} |p_{ijk}|^2} \\
        = &~ d_{A} \sqrt{ \sum_{k=1}^{d_{A}^2} \bra{\alpha_i}\bra{\overline{\alpha_j}} F^{-1} \kketbbra{\psi_k} F^{-1} \ket{\alpha_i}\ket{\overline{\alpha_j}} } \\
        = &~ d_{A} \sqrt{ \bra{\alpha_i}\bra{\overline{\alpha_j}} F^{-1} \ket{\alpha_i}\ket{\overline{\alpha_j}} } \\
        \leq &~ d_{A} \lambda_{\min}^{-1/2} \,.
    \end{align}
    Last,
    \begin{align}
         &~ \| (\map C \otimes \map I_{A'})(\ketbra{\Phi}) - (\map D \otimes \map I_{A'})(\ketbra{\Phi}) \|_1 \\
        =&~ \left\| \sum_{i,j=1}^r \sum_{k=1}^{d_{A}^2} s_i s_j p_{ijk} (\map C - \map D)({\psi_k}) \otimes \ket{\beta_i}\bra{\beta_j} \right\|_1 \\
        =&~ \left\| \sum_{i,j=1}^r \sum_{k=1}^{d_{A}^2} s_i s_j p_{ijk} (\rho_k - \rho_1) \otimes \ket{\beta_i}\bra{\beta_j} \right\|_1 \\
        \leq &~  \sum_{i,j=1}^r \sum_{k=1}^{d_{A}^2} s_i s_j |p_{ijk}| \left\| (\rho_k - \rho_1) \otimes \ket{\beta_i}\bra{\beta_j} \right\|_1 \\
        \leq &~  \sum_{i}^r s_i \sum_{j}^r s_j \sum_{k=1}^{d_{A}^2} |p_{ijk}| \left\| \rho_k - \rho_1 \right\|_1 \\
        \leq &~  \sqrt{d_{A}\sum_{i}^r s_i^2}\sqrt{d_{A}\sum_{j}^r s_j^2} ~ \left( d_{A}\lambda_{\min}^{-1/2} \right) \varepsilon \\
        \leq &~  d_{A}^2  \lambda_{\min}^{-1/2}  \varepsilon
    \end{align}
\end{proof}

\begin{lem}\label{lem:2norm}
Let $\map{C}: S(\map{H}_{A_1 A_2}) \to S(\map{H}_{B_1 B_2})$ be a channel with Choi state $C$. Let $\{\psi_b\}_{b=1}^{d_{A_2}^2}$ be an informationally complete set of states for system $A_2$.
Let $\rho_b := \Tr_{B_2}[ (\map{C} \otimes I_{A'_1}) (\ketbra{\Phi^+}_{A_1 A'_1} \otimes \psi_{b,A_2})]$. If $\forall b, ~ \|\rho_b - \rho_1\|_2 \leq \varepsilon$
, then there exists $\map D \in \set{Comb}[(A_1,B_1),(A_2, \emptyset )]$ such that
\begin{align}
    \|\Tr_{B_2}[C] - D\|_1 \leq \sqrt{2 d_{B_2}\rank(C)} ~ d_{A_2}^2 \lambda_{\min}^{-1/2}  \varepsilon \,,
\end{align}
where $D$ is the Choi state of $\map{D}$ and $\set{Comb}[(A_1,B_1),(A_2,\emptyset)]$ denotes combs whose second tooth does not have an output.
\end{lem}
\begin{proof}
Defining channel $\map{C}': \spc{H}_{A_2} \to \spc{H}_{A_1'}\otimes\spc{H}_{B_1}$ as $\map{C}'(\rho) := \Tr_{B_2}[ (\map{C} \otimes I_{A'_1}) (\ketbra{\Phi^+}_{A_1 A'_1} \otimes \rho )]$, we have $\rho_b = \map{C}'(\psi_b)$ and $\rank(\map{C}') \leq d_{B_2} \rank(\map{C})$.
According to Ref. \cite{coles2019strong}, we have $\|\rho_b - \rho_1\|_1^2 \leq (\rank(\rho_b)+\rank(\rho_1)) \|\rho_b - \rho_1\|_2^2 \leq 2\rank(\map{C}') \varepsilon^2 \leq 2d_{B_2}\rank(\map{C}) \varepsilon^2$, and $\|\rho_b - \rho_1\|_1 \leq \sqrt{2d_{B_2}\rank(C)} \, \varepsilon$. According to \autoref{lem:ic_diamond}, there exists a constant map $\map D': S(\spc{H}_{A_2}) \to S(\spc{H}_{A_1'}\otimes \spc{H}_{B_1})$, $\map{D}'(\rho) := \rho_1 \Tr[\rho]$ such that
\begin{align}
    \|\map C' - \map D'\|_\diamond \leq d_{A_2}^2 \lambda_{\min}^{-1/2}  \sqrt{2 d_{B_2}\rank(C)} \, \varepsilon \,.
\end{align}
Choose $D=\rho_1\otimes I_{A_2}$ as the Choi state of $\map D'$. By choosing a proper  $\map D \in \set{Comb}[(A_1,B_1),(A_2,\emptyset )]$, $D$ is alsothe  Choi state for $\map{D}$ since $\Tr_{B_1}[D] = \Tr_{B_1}[\rho_1 \otimes I_{A_2}] = I_{A'_1} \otimes I_{A_2}$. By observing that $\Tr_{B_2}[C]$ is the Choi state of $\map C'$, we arrive at the following
\begin{align}
    \|\Tr_{B_2}[C] - D\|_1 \leq \|\map C' - \map D'\|_\diamond \leq \sqrt{2 d_{B_2}\rank(C)} ~ d_{A_2}^2 \lambda_{\min}^{-1/2}  \varepsilon \,.
\end{align}
\end{proof}

\begin{lem} \label{lem:p-tr}
For quantum states $\rho$ and $\sigma$ in Hilbert space $\spc{H}_A$, if $\|\rho - \sigma\|_1 \leq \varepsilon$, then for any purification of $\rho$, $\ket{\psi}_{AR}$ with a ancillary system $R$ satisfying $\dim R \geq \rank(\sigma)$, there exists a purification of $\sigma$, $\ket{\phi}_{AR}$ such that
\begin{align}
    \| \ketbra{\psi} - \ketbra{\phi} \|_1 \leq 2\sqrt\varepsilon
\end{align}
\end{lem}
\begin{proof}
By Uhlmann's Theorem, there exists a purification of $\sigma$, $\ket{\phi}_{AR}$ such that $F(\rho,\sigma) = F(\ketbra{\psi}, \ketbra{\phi})$.
\begin{align}
    \frac12 \| \ketbra{\psi} - \ketbra{\phi} \|_1 \leq \sqrt{1-F(\ketbra{\psi}, \ketbra{\phi})} = \sqrt{1-F(\rho,\sigma)} \leq \sqrt{1-\left(1-\frac12\|\rho - \sigma\|_1\right)^2} \leq \sqrt{\varepsilon - \varepsilon^2/4} \leq \sqrt{\varepsilon}
\end{align}
\end{proof}

\begin{lem}\label{lem:2ton}
Let $\map{C}$ be a quantum channel and $C$ be its Choi state. Let $C^{(i)}:= \Tr_{A_{i+1}B_{i+1}\dots A_{n}B_{n}}[C]$ be the Choi operator on first $i$ input-output pairs. For $n\geq 2$, if
\begin{align} \label{eq:CDi}
    \forall i=2,\dots,n,~ \exists \map{D}^{(i)\circ} \in \set{Comb}[(A_1\dots A_{i-1},B_1\dots B_{i-1}), (A_i, \emptyset )], ~ \|C^{(i)\circ} - D^{(i)\circ}\|_1 \leq \delta \,,
\end{align}
where $C^{(i)\circ}:= \Tr_{B_i}[C^{(i)}]$, then
\begin{align}
    \exists \map{D} \in \set{Comb}[(A_1,B_1), \dots , (A_n, B_n)], ~ \|C - D\|_1 \leq (4n-6)\sqrt\delta \,,
\end{align}
where $D$ is the Choi state of $\map{D}$.
\end{lem}
\begin{proof}~

In this proof, calligraphic letters denote linear maps and italic letters denote the corresponding Choi states.

Define a proposition $P(i)$: there exists a purification of $C^{(i)}$, denoted as $C^{(i)*} \in S(\spc{H}_{A_1 \dots A_{i}} \otimes \spc{H}_{B_1 \dots B_{i}} \otimes \spc{H}_{M_i})$, and an isometric channel $\map E^{(i)*} \in \set{Comb}[(A_1,B_1),\dots, (A_{i-1}, B_{i-1}), (A_i, B_iM_i)]$ such that
\begin{align}
    \|C^{(i)*} - E^{(i)*} \|_1 \leq \varepsilon^{(i)} \,.
\end{align}

Consider $P(2)$. Pick any purification of $C^{(2)}$ as $C^{(2)*}$, and $C^{(2)*}$ is also a purification of $C^{(2)\circ}$. Pick a purification of $D^{(2)\circ}$ as $D^{(2)*}$ such that its trace distance to $C^{(2)*}$ is minimized. Since $\|C^{(2)\circ} - D^{(2)\circ}\|_1 \leq \delta$, according to \autoref{lem:p-tr}, we have $\|C^{(2)*} - D^{(2)*}\|_1 \leq 2\sqrt\delta$. By choosing $E^{(2)*}=D^{(2)*}$, we prove $P(2)$ with $\varepsilon^{(2)} = 2\sqrt{\delta}$.

Assume $P(i-1)$ holds with $\varepsilon^{(i-1)}$, then there exists a purification of $C^{(i-1)}$, denoted as $C^{(i-1)*} \in S(\spc{H}_{A_1 \dots A_{i-1}} \otimes \spc{H}_{B_1 \dots B_{i-1}} \otimes \spc{H}_{M_{i-1}})$, and an isometric channel $\map E^{(i-1)*} \in \set{Comb}[(A_1,B_1),\dots, (A_{i-2}, B_{i-2}), (A_{i-1}, B_{i-1}M_{i-1})]$ such that
\begin{align}\label{eq:CEminus1}
    \|C^{(i-1)*} - E^{(i-1)*} \|_1 \leq \varepsilon^{(i-1)} \,.
\end{align}

Pick a purification of $C^{(i)}$ as $C^{(i)*} \in S(\spc{H}_{A_1 \dots A_{i}} \otimes \spc{H}_{B_1 \dots B_{i}} \otimes \spc{H}_{M_{i}})$ where $\spc{H}_{M_{i}}$ is a large enough Hilbert space. Then $C^{(i)*}$ is also a purification of $C^{(i)\circ}$. Pick a purification of $D^{(i)\circ}$ as $D^{(i)*}$ such that its trace distance to $C^{(i)*}$ is minimized. Since $\|C^{(i)\circ} - D^{(i)\circ}\|_1 \leq \delta$ (\ref{eq:CDi}), according to \autoref{lem:p-tr}, we have
\begin{align} \label{eq:CDstar}
\| C^{(i)*} - D^{(i)*} \|_1 \leq 2\sqrt{\|C^{(i)\circ} - D^{(i)\circ}\|_1} \leq 2\sqrt\delta \,.
\end{align}

Let $(D')^{(i-1)} := \Tr_{A_i}[ D^{(i)\circ} ]$. Pick a purification of $(D')^{(i-1)}$ as $(D')^{(i-1)*}\in S(\spc{H}_{A_1 \dots A_{i-1}} \otimes \spc{H}_{B_1 \dots B_{i-1}} \otimes \spc{H}_{M_{i-1}})$  such that its trace distance to $C^{(i-1)*}$ is minimized. From \autoref{lem:p-tr} we have
\begin{align}\label{eq:DprimeC}
    \|(D')^{(i-1)*} - C^{(i-1)*}\|_1 \leq 2\sqrt{\|(D')^{(i-1)} - C^{(i-1)}\|_1 } = 2\sqrt{\| \Tr_{A_i}[D^{(i)\circ}-C^{(i)\circ}] \|_1 } \leq 2\sqrt{\| D^{(i)\circ} - C^{(i)\circ} \|_1 } \leq 2 \sqrt\delta
\end{align}

Note that $\map D^{(i)*} \in \set{Comb}[(A_1\dots A_{i-1},B_1\dots B_{i-1}), (A_i, B_iM_i)]$ since $\map{D}^{(i)\circ} \in \set{Comb}[(A_1\dots A_{i-1},B_1\dots B_{i-1}), (A_i, \emptyset )]$. Therefore $\map D^{(i)*}$ can be written as a concatenation of two channels with the first channel being $(D')^{(i-1)*}$:

\begin{align}
    \map D^{(i)*} =&~ (\map I_{B_1\dots B_{i-1}} \otimes \map{D}_i) \circ ((\map D')^{(i-1)*} \otimes \map I_{A_i})  \label{eq:Distar}     
\end{align}
where $\map D_i: S(\spc{H}_{A_i}\otimes\spc{H}_{M_{i-1}}) \to S(\spc{H}_{B_i}\otimes\spc{H}_{M_i})$ is an isometric channel since $D^{(i)*}$ is isometric.

Define
\begin{align}
    \map E^{(i)*} := (\map I_{B_1\dots B_{i-1}} \otimes \map{D}_i) \circ (\map E^{(i-1)*} \otimes \map I_{A_i})
\end{align}
and let $E^{(i)*}$ be its Choi state. This is an isometric channel because both $\map{D}_i$ and $\map E^{(i-1)*}$ are. Define $\tilde{\map{D}}_i(\rho): S(\spc{H}_{M_{i-1}}) \to S(\spc{H}_{A'_i}\otimes \spc{H}_{B_i} \otimes \spc{H}_{M_i})$ as $\tilde{\map{D}}_i(\rho) :=  (\map{D}_i \otimes \map I_{A'_i})( \rho \otimes \ketbra{\Phi^+}_{A_i A'_i})$, which is an isometric channel and preserves trace norm. Then we have
\begin{align}
\nonumber \| D^{(i)*} - E^{(i)*} \|_1 & = \left\| (\tilde{\map{D}}_i \otimes \map I_{A'_1 \dots A'_i} \otimes \map I_{B_1 \dots B_{i-1}}) \left( (D')^{(i-1)*} - E^{(i-1)*} \right) \right\|_1 \\
\nonumber    & = \| (D')^{(i-1)*} - E^{(i-1)*} \|_1 \\
\nonumber    & \leq \| (D')^{(i-1)*} - C^{(i-1)*} \|_1 + \| C^{(i-1)*} - E^{(i-1)*} \|_1 \\
    & \leq 2\sqrt\delta + \varepsilon^{(i-1)} \label{eq:DEstar}
\end{align}
where the last step uses Eqs. (\ref{eq:DprimeC}) and (\ref{eq:CEminus1}).

Last, from Eqs. (\ref{eq:CDstar}) and (\ref{eq:DEstar}),
\begin{align}
\| C^{(i)*} - E^{(i)*} \|_1 \leq \| C^{(i)*} - D^{(i)*} \|_1 + \| D^{(i)*} - E^{(i)*} \|_1 \leq 4\sqrt\delta + \varepsilon^{(i-1)}
\end{align}
Therefore, we can choose $\varepsilon^{(i)} = \varepsilon^{(i-1)} + 4\sqrt\delta$. Performing induction with base case $\varepsilon^{(2)} = 2\sqrt\delta$, we have
\begin{align}
    \varepsilon^{(n)} = (4n-6)\sqrt\delta
\end{align}
Picking $D := \Tr_{M_n}[ E^{(n)*}]$, we have
\begin{align}
    \| C - D \|_1 \leq \| C^{(n)*} - E^{(n)*} \|_1 \leq \varepsilon^{(n)} = (4n-6)\sqrt\delta
\end{align}
\end{proof}

\begin{lem}\label{lem:rank}
For a constant channel $\map{C}: S(\spc{H}_{A})\to S(\spc{H}_B)$ satisfying $\map{C}(X) = \Tr[X] \sigma_B$, one has $\rank(\sigma_B) =\rank(\map{C})/ \dim\spc{H}_A$, where $\rank(\map{C})$ is the Kraus rank of channel $\map{C}$.
\end{lem}
\begin{proof}
   Since $\map{C}(X) = \Tr[X] \sigma_B$, the Choi state of $\map{C}$ is $C = I_A/d_A \otimes \sigma_B$. Then $\rank(C) = \rank(I_A)\rank(\sigma_B)$, and $\rank(\sigma_B) = \rank(C) / \rank(I_A) = \rank(\map{C})/ \dim\spc{H}_A$.
\end{proof}
\begin{lem}\label{lem:rank2}
Under \autoref{ass:unitary}, assuming $\dim A_i = \dim B_j, \forall i,j$, conditioned on that each independence test asserts the correct answer, during the execution of \autoref{alg:unitary}, it always holds true that $\rank(\map{O}) \leq \rank(\map{C}) \leq d_M$.
\end{lem}
\begin{proof}
    At the beginning of \autoref{alg:unitary}, $\rank(\map{O}) \leq \rank(\map{C})$ trivially holds since $\map{O}=\map{C}$. Now we show that in each iteration, updating the definition of $\map{O}$ does not increase its rank. Let $\map{O}^{(k)}$ be the value of $\map{O}$ at the beginning of iteration $k$. Assume $\rank(\map{O}^{(k)}) \leq \rank(\map{C})$ and consider $\map{O}^{(k-1)}$ which is the value after the update in iteration $k$.

    Conditioned on that each independence test produces the correct answer, \autoref{func:findlast} produces a pair of wires $(A_x,B_y)$ satisfying that the following channel $\map{P}: \spc{H}_{A_x} \to \spc{H}_{B_{\neq y}} \otimes \spc{H}_{A'_{\neq x}}$ is a constant channel:
    $$\map{P}(\rho_{A_x}) := \Tr_{B_y} \left[ \left( \map{O}^{(k)}\otimes\map{I_{A'_{\neq x}}} \right) \left( \bigotimes_{k \neq x} \frac{\kketbbra{I}_{A_kA'_k}}{d_{A_k}} \otimes \rho_{A_x} \right) \right]$$
    \autoref{fig:rhob} shows the circuit of $\map{P}$ with input $\psi_b$.

    Since $\map{P}$ is a constant channel, $\map{P}(X) = \Tr[X] O^{(k-1)}$, where $O^{(k-1)}$ equals to the Choi state of the updated channel $\map{O}^{(k-1)}$. According to \autoref{lem:rank}, $\rank(\map{O}^{(k-1)})=\rank(O^{(k-1)}) = \rank(\map{P})/ \dim\spc{H}_{A_x} \leq \rank(\map{O}^{(k)})\dim\spc{H}_{B_y}/ \dim\spc{H}_{A_x} =  \rank(\map{O}^{(k)})\leq \rank(\map{C})$, where we used $\dim\spc{H}_{A_x} = \dim\spc{H}_{B_y} = d_A$ in \autoref{ass:unitary}. Since $\map{C}$ is composed of only unitary gates with a $d_M$-dimensional system traced out, we have $\rank(\map{C}) \leq d_M$.
\end{proof}

Now we prove \autoref{thm:unitary}.
\begin{proof}[Proof of \autoref{thm:unitary}]
We work under the condition that all SWAP tests are successful, namely all estimates the SWAP test produces are within error $\varepsilon$, which has probability no less than $1-\kappa$ according to \autoref{lem:SWAP}. There are at most $2 n d_A^2$ SWAP tests, and the probability that all of them succeeds is no less than $1- 2nd_A^2\kappa$. Let $\kappa = \kappa_0/2nd_A^2$, and from now we assume all SWAP tests are successful, which has probability no less than $1-\kappa_0$.

If \autoref{func:findlast} outputs $(i,j)$, we know that for all $k$, $p_{1} + p_{k} - 2 p_{1k} \leq \delta$. Since $p_{1} + p_{k} - 2 p_{1k}$ is an estimate for $\Tr[\rho_1^2]+\Tr[\rho_k^2]-2\Tr[\rho_1\rho_k] = \|\rho_1 - \rho_k\|_2^2$ up to error $4\varepsilon$, we infer that  $\|\rho_1 - \rho_k \|_2^2 \leq  \delta + 4\varepsilon$.

According to \autoref{ass:unitary}, in \autoref{func:findlast}, if $(A_i,B_j)$  is the last tooth, we must have $\|\rho_1 - \rho_k \|_2=0, \forall k$. By setting $\varepsilon = \delta/4$, for the last tooth, we can ensure that  $p_{1} + p_{k} - 2 p_{1k} \leq \|\rho_1 - \rho_k \|_2^2 + 4\varepsilon = \delta$, the last tooth must pass the independence test, and \autoref{func:findlast} will not fail.

Let $\varepsilon = \delta/4$, we have $\|\rho_1 - \rho_k \|_2 \leq \sqrt{2\delta}$. Let $O$ be the Choi state of $\map{O}$. According to \autoref{lem:2norm}, we have
\begin{align}
    \|\Tr_{B_j}[O] - D\|_1 \leq \sqrt{2 d_{B_j} \rank(O)} ~ d_{A_i}^2 \lambda_{\min}^{-1/2} \sqrt{2\delta} \leq 2\sqrt{d_M \delta} ~ d_{A}^{5/2} \lambda_{\min}^{-1/2} \,,
\end{align}
where $\map{D} \in \set{Comb}[(A_{\neq i},B_{\neq_j}),(A_i,\emptyset)]$ and $A_{\neq i}$ ($B_{\neq_j}$) is the set of input (output) wires of $\map{O}$ excluding $A_i$ ($B_j$). The last inequality follows from \autoref{lem:rank2}, which asserts $\rank(O) \leq d_M$.

According to \autoref{lem:2ton}, we have
\begin{align}
    \exists \map{D} \in \set{Comb}[(A_{\sigma(1)},B_{\pi(1)}),\dots,(A_{\sigma(n)},B_{\pi(n)})], ~ \|C - D\|_1 \leq (4n-6)\sqrt{2\sqrt{2 d_M} d_A^{5/2} \lambda_{\min}^{-1/2} \delta} \,.
\end{align}
The right hand side becomes $\varepsilon_0$ if we choose
\begin{equation}
    \delta = \frac{ \lambda_{\min} \varepsilon_0^4}{4(4n-6)^4 d_M d_A^{5}}
\end{equation}

Now we analyze the complexity of \autoref{alg:unitary}. \autoref{func:findlast} is called $n$ times, and each time it invokes $O(n^2 d_A^2)$ number of SWAP tests. Each SWAP test runs the circuit $N$ times, which is
\begin{align}
    N = \lceil 2\varepsilon^{-2}\log(2/\kappa) \rceil = O\left(n^8 d_A^{10} d_M^2 \varepsilon_0^{-8} \lambda_{\min}^{-2} \log (n d_A \kappa_0^{-1})\right) \,.
\end{align}
Since there are $O(n^3 d_A^2)$ number of SWAP tests, the total query complexity is
\begin{align}T_{\rm query}=O\left(n^{11} d_A^{12} d_M^2 \varepsilon_0^{-8} \lambda_{\min}^{-2} \log (n d_A \kappa_0^{-1})\right)\end{align}

The processing time is also dominated by the SWAP tests. Assuming the input and output systems are represented with qubits, each wire requires $O(\log d_A)$ qubits and the inputs and outputs total to $O(n\log d_A)$ qubits.
Each run of the circuit in \autoref{fig:SWAP} requires applying a controlled-SWAP gate on $O(n)$ wires, which could be decomposed as $O(n\log d_A)$ controlled-SWAP gates between qubits. Assuming each qubit controlled-SWAP gate takes constant time, the total processing time is $O(T_{\rm query}n\log d_A)$.

\end{proof}

\section{Proof of \autoref{lem:alg_ind}}\label{app:alg_ind}


\begin{lem} \label{lem:2normbounds}
    Let $\{P_\alpha\}$ be an informationally complete POVM and $F :=\sum_\alpha \kketbbra{P_\alpha}$ be its frame operator. For two Hermitian operators $\rho$ and $\sigma$, define $p_\alpha:=\Tr[P_\alpha \rho]$ and $q_\alpha:=\Tr[P_\alpha \sigma]$. Then one has
    \begin{align} \label{eq:2normbounds}
        \frac{\sum_{\alpha}(p_\alpha-q_\alpha)^2}{\lambda_{\max}(F)} \leq \|\rho - \sigma\|_2^2 \leq \frac{\sum_{\alpha}(p_\alpha-q_\alpha)^2}{\lambda_{\min}(F)} \,,
    \end{align}
    where $\|X\|_2:= \sqrt{\Tr[X^\dag X]}$ denotes the Schatten 2-norm, also known as the Frobenius norm, and $\lambda_{\max}(F)$ and $\lambda_{\min}(F)$ are the maximum and minimum eigenvalues of $F$, respectively.
\end{lem}
\begin{proof}
    Let $T := \rho - \sigma$ and $t_{\alpha} := p_\alpha - q_\alpha$. Then $t_\alpha = \Tr[P_\alpha \rho]-\Tr[P_\alpha \sigma] = \Tr[P_\alpha T]$. Define the vector $\ket{t} := \sum_\alpha t_\alpha \ket{\alpha}$, and operator $S := \sum_\alpha \ket{\alpha}\bbra{P_\alpha}$, where $\{\ket\alpha\}$ is a set of orthonormal vectors indexed by $\alpha$. Then one has $\ket{t} = S\kket{T}$ and $S^\dag S = F$.
    \begin{align}
        \sum_\alpha t_\alpha^2 =  \braket{t}{t} = \bbra{T} S^\dag S \kket{T} = \bbra{T} F \kket{T}\,.
    \end{align}
    Since $\lambda_{\min}(F)\bbrakket{T}{T} \leq \bbra{T} F \kket{T} \leq \lambda_{\max}(F)\bbrakket{T}{T}$, we obtain
    \begin{align}
        \lambda_{\min}(F)\bbrakket{T}{T} \leq \sum_\alpha t_\alpha^2 \leq \lambda_{\max}(F)\bbrakket{T}{T} \,,
    \end{align}
    which is equivalent to Eq. (\ref{eq:2normbounds}) since $\bbrakket{T}{T} = \|T\|_2^2$.
\end{proof}

\begin{proof}[Proof of \autoref{lem:alg_ind}]
Let $\hat\rho_A := \Tr_B[\hat\rho_{AB}]$ and $\hat\rho_B := \Tr_A[\hat\rho_{AB}]$ be the marginal states of the reconstructed state $\hat\rho_{AB}$.
Define $\tau := \rho_{AB} - \rho_A \otimes \rho_B$ and $\hat\tau := \hat\rho_{AB} - \hat\rho_A \otimes \hat\rho_B$, the error of \autoref{alg:c1} is then the difference between $\|\tau\|_1=\chi_1(\rho_{AB})$ and $\|\hat{\tau}\|_1=\chi_1(\hat\rho_{AB})$. 

Let $p_{\alpha\beta}:= \Tr[(P_\alpha \otimes Q_\beta)\rho_{AB}]$, $p^A_{\alpha}:=\sum_\beta p_{\alpha\beta}$ and $p^B_{\beta} = \sum_\alpha p_{\alpha\beta}$. Let $\hat{p}^A_{\alpha}:=\sum_\beta \hat{p}_{\alpha\beta}$ and $\hat{p}^B_{\beta} = \sum_\alpha \hat{p}_{\alpha\beta}$. Over the $N$ measurement outcomes, $\hat{p}_{\alpha\beta}$ is the average of independent Bernouli random variables, and has mean $p_{\alpha\beta}$. The same is true for $\hat{p}^A_{\alpha}$ and $\hat{p}^B_{\beta}$, whose means are $p^A_{\alpha}$ and $p^B_{\beta}$, respectively.

Let $\varepsilon_1:= \xi\varepsilon$. By Hoeffding's inequality,
\begin{align}
    \Pr(|p_{\alpha\beta} - \hat{p}_{\alpha\beta}| \leq \varepsilon_1) &\geq 1-2 e^{-2\varepsilon_1^2 N} \\
    \Pr(|p^A_{\alpha} - \hat{p}^A_{\alpha}| \leq \varepsilon_1) &\geq 1-2 e^{-2\varepsilon_1^2 N}\\
    \Pr(|p^B_{\beta} - \hat{p}^B_{\beta}| \leq \varepsilon_1) &\geq 1-2 e^{-2\varepsilon_1^2 N}
\end{align}
By the union bound,
\begin{align}
    \Pr\left[ \left(\forall\alpha,\beta, |p_{\alpha\beta} - \hat{p}_{\alpha\beta}| \leq \varepsilon_1\right) \wedge \left(\forall \alpha, |p^A_{\alpha} - \hat{p}^A_{\alpha}| \leq  \varepsilon_1\right) \wedge \left(\forall \beta, |p^B_{\beta} - \hat{p}^B_{\beta}| \leq \varepsilon_1\right)  \right] \geq 1-2 (d_A^2d_B^2+d_A^2+d_B^2) e^{-2\varepsilon_1^2 N},
\end{align}
namely with probability no less than $1-2 (d_A^2d_B^2+d_A^2+d_B^2) e^{-2\varepsilon_1^2 N} = 1-\kappa(\varepsilon)$, all $(d_A^2d_B^2+d_A^2+d_B^2)$ number of the following inequalities
\begin{align}
    |p_{\alpha\beta} - \hat{p}_{\alpha\beta}| \leq & \varepsilon_1, ~\forall \alpha,\beta \label{eq:phat1} \\
    |p^A_{\alpha} - \hat{p}^A_{\alpha}| \leq & \varepsilon_1, ~\forall \alpha\\
    |p^B_{\beta} - \hat{p}^B_{\beta}| \leq & \varepsilon_1, ~\forall \beta \label{eq:phat3}
\end{align}
are satisfied. We then analyze the estimation error of $\chi_1(\rho_{AB})$ in this case. Define $t_{\alpha\beta}:=p_{\alpha\beta} - p^A_{\alpha}p^B_{\beta}$. By Eqs. (\ref{eq:phat1}-\ref{eq:phat3}),
\begin{align}
|t_{\alpha\beta} - \hat{t}_{\alpha\beta}|  & = |p_{\alpha\beta} - \hat{p}_{\alpha\beta} - p^A_{\alpha}p^B_{\beta} + \hat{p}^A_{\alpha}\hat{p}^B_{\beta}| \\
&\leq  |p_{\alpha\beta} - \hat{p}_{\alpha\beta}| + |p^A_{\alpha}p^B_{\beta} - p^A_{\alpha}\hat{p}^B_{\beta}|  + |p^A_{\alpha}\hat{p}^B_{\beta} - \hat{p}^A_{\alpha}\hat{p}^B_{\beta}| \\
&\leq  \varepsilon_1 + p^A_{\alpha}|p^B_{\beta} -\hat{p}^B_{\beta}|  + |p^A_{\alpha} - \hat{p}^A_{\alpha}| \hat{p}^B_{\beta} \\
&\leq  (1+p^A_{\alpha} + \hat{p}^B_{\beta})\varepsilon_1
\end{align}

Since $\sum_\alpha p^A_{\alpha} = \sum_\beta \hat{p}^B_{\beta} = 1$, we have
\begin{align}
\sum_{\alpha,\beta} (t_{\alpha\beta} - \hat{t}_{\alpha\beta})^2 &\leq \sum_{\alpha,\beta} (1+p^A_{\alpha} + \hat{p}^B_{\beta})^2\varepsilon_1^2\\
&\leq  \sum_{\alpha,\beta} (1+4p^A_{\alpha} + 4\hat{p}^B_{\beta}) \varepsilon_1^2\\
& = (d_A^2 d_B^2 + 4d_B^2 + 4d_A^2) \varepsilon_1^2
\end{align}

Note that $t_{\alpha\beta} = \Tr[(P_\alpha\otimes Q_\beta)\tau]$ and $\hat{t}_{\alpha\beta} = \Tr[(P_\alpha\otimes Q_\beta)\hat{\tau}]$. Applying \autoref{lem:2normbounds} to POVM $\{P_\alpha\otimes Q_\beta\}$ and operators $\tau$ and $\hat{\tau}$, we obtain
\begin{align}
\|\tau - \hat{\tau}\|_2^2 \leq \frac{\sum_{\alpha,\beta}\left(t_{\alpha\beta}-\hat{t}_{\alpha\beta}\right)^2}{\lambda_{\min}(F\otimes G)} \leq \frac{(d_A^2 d_B^2 + 4d_B^2 + 4d_A^2) \varepsilon_1^2}{\lambda_{\min}(F)\lambda_{\min}(G)}
\end{align}

Using Eq. (\ref{eq:norm12}), we have
\begin{align}
\|\tau - \hat{\tau}\|_1 \leq \sqrt{\rank(\tau - \hat{\tau})}\|\tau - \hat{\tau}\|_2 \leq d_Ad_B \sqrt{\frac{d_A^2 d_B^2 + 4d_B^2 + 4d_A^2 }{\lambda_{\min}(F)\lambda_{\min}(G)}} \, \varepsilon_1 = \varepsilon
\end{align}
which proves Eq. (\ref{eq:c1_bound}) since $| \|\tau\|_1 - \|\hat{\tau}\|_1 | \leq \|\tau - \hat\tau\|_1$.
\end{proof}

\section{Proof of \autoref{thm:memoryless}} \label{app:memoryless}

Picking $\chi_- = \varepsilon_0$, according to \autoref{lem:ind}, with probability $1-n^2 \kappa_0$,
\begin{enumerate}
    \item if $\ind_{i,j} = \False$, then $\chi_1(C_{A_i,B_j}) > 0$
    \item if $\ind_{i,j} = \True$, then $\chi_1(C_{A_i,B_j}) \leq 2\chi_-$
\end{enumerate}
Since by \autoref{ass:memoryless}, for every input $A_i$, there is at most one output $B_j$ such that $\chi_1(C_{A_i,B_j}) > 0$, and therefore there is at most one $j$ satisfying $\ind_{i,j}=\False$. Let $A_{\rm matched} := \{A_i | \exists j, \ind_{i,j}=\False \}$ and $B_{\rm matched}:=\{ B_{\pi(i)} | A_i \in A_{\rm matched}\}$ be the inputs and outputs that are matched up to line \ref{line:memoryless_match} of \autoref{alg:memoryless}. 

Define the channel $\map{D}$ by its action on product states as follows:
\begin{align}
    \map{D}(\rho_{A_1} \otimes \dots \otimes \rho_{A_n}) :=  \bigotimes_{i \in A_{\rm matched}} \map{C}_{A_i \to B_{\pi(i) }}(\rho_{A_i}) \otimes   \bigotimes_{i \notin A_{\rm matched}}  \map{C}_{A_i \nrightarrow B_{\pi(i)} }(\rho_{A_i})
\end{align}
where $\map{C}_{A_i \to B_{\pi(i) }}(\rho_{A_i}) := \Tr_{B_{\neq \pi(i)}}[\map{C}(\rho_{A_i} \otimes I_{A_{\neq i}}/d_{A_{\neq i}})]$ is channel $\map{C}C$ restricted on input $A_i$ and output $B_{\pi(i)}$, and $\map{C}_{A_i \nrightarrow B_{\pi(i)} }(\rho_{A_i}) := \Tr[\rho_{A_i}] C_{B_{\pi(i)}}$ is a constant channel that outputs $C_{B_{\pi(i)}}$, the marginal Choi state of $\map{C}$ on system $B_{\pi(i)}$.

Clearly, $\map{D} \in \set{Comb}[ (A_{1}, B_{\pi(1)}), \dots , (A_{n}, B_{\pi(n)}) ]$, since it is a tensor product of channels from $A_{i}$ to $B_{\pi(i)}$. Now we show that the difference between $\map{C}$ and $\map{D}$ is small.

According to \autoref{ass:memoryless}, $\map{C} = \bigotimes_{i=1}^n \map{C}_{A_i \to B_{\pi''(i)}} = \bigotimes_{i \in A_{\rm matched}} \map{C}_{A_i \to B_{\pi''(i) }} \otimes   \bigotimes_{i \notin A_{\rm matched}}  \map{C}_{A_i \to B_{\pi''(i) }}$.
For channel $\map{C}$ and each $i \in A_{\rm matched}$, $A_{i}$ is correlated to $B_{\pi(i)}$ but not correlated to all other outputs, and by the way $\pi$ is computed, one must have $\pi(i)=\pi''(i)$. Therefore,
\begin{align}
    \| \map{C} - \map{D} \|_\diamond & = \left\|\bigotimes_{i \in A_{\rm matched}} \map{C}_{A_i \to B_{\pi''(i) }} \otimes   \bigotimes_{i \notin A_{\rm matched}}  \map{C}_{A_i \to B_{\pi''(i) }} -  \bigotimes_{i \in A_{\rm matched}} \map{C}_{A_i \to B_{\pi(i) }} \otimes   \bigotimes_{i \notin A_{\rm matched}} \map{C}_{A_i \nrightarrow B_{\pi(i)} } \right\|_\diamond \\
    & = \left\|\bigotimes_{i \notin A_{\rm matched}}  \map{C}_{A_i \to B_{\pi''(i) }}  -  \bigotimes_{i \notin A_{\rm matched}} \map{C}_{A_i \nrightarrow B_{\pi(i)} }  \right\|_\diamond \label{eq:CABCAB}
\end{align}
and the difference between $\map{C}$ and $\map{D}$ only occurs in the second part involving $i \notin A_{\rm matched}$. Since
\begin{align}
    \bigotimes_{i \notin A_{\rm matched}}  \Tr[\rho_{A_i}] C_{B_{\pi(i)}} = \prod_{i \notin A_{\rm matched}}  \Tr[\rho_{A_i}] \bigotimes_{j \notin B_{\rm matched}} C_{B_j} = \bigotimes_{i \notin A_{\rm matched}}  \Tr[\rho_{A_i}] C_{B_{\pi''(i)}}
\end{align}
we have $\bigotimes_{i \notin A_{\rm matched}}  \map{C}_{A_i \nrightarrow B_{\pi(i) }} = \bigotimes_{i \notin A_{\rm matched}}  \map{C}_{A_i \nrightarrow B_{\pi''(i) }}$, where $\map{C}_{A_i \nrightarrow B_{\pi''(i) }}$ is the constant channel defined as $\map{C}_{A_i \nrightarrow B_{\pi''(i)} }(\rho_{A_i}) := \Tr[\rho_{A_i}] C_{B_{\pi''(i)}}$. Eq. (\ref{eq:CABCAB}) then becomes
\begin{align} \label{eq:CDdia}
    \| \map{C} - \map{D} \|_\diamond & = \left\|\bigotimes_{i \notin A_{\rm matched}}  \map{C}_{A_i \to B_{\pi''(i) }}  -  \bigotimes_{i \notin A_{\rm matched}} \map{C}_{A_i \nrightarrow B_{\pi''(i)} }  \right\|_\diamond \leq \sum_{i \notin A_{\rm matched}} \| \map{C}_{A_i \to B_{\pi''(i) }} - \map{C}_{A_i \nrightarrow B_{\pi''(i)} } \|_\diamond
\end{align}


Consider the channel $\map{C}_{A_i \to B_{\pi''(i) }}$. For $i \notin A_{\rm matched}$, $\ind_{i,\pi''(i)} = \True$ and with high probability $\chi_1(C_{A_i,B_{\pi''(i)}}) \leq 2\chi_-$. By definition of $\chi_1$,
\begin{align}\label{eq:diff_choi_leq_chi}
    \| C_{A_i,B_{\pi''(i)}} - C_{A_i}\otimes C_{B_{\pi''(i)}} \|_1 = \left\| C_{A_i,B_{\pi''(i)}} - \frac{I_{A_i}}{d_{A_i}}\otimes C_{B_{\pi''(i)}} \right\|_1 \leq 2\chi_-
\end{align}
Eq. (\ref{eq:diff_choi_leq_chi}) is the distance between the Choi states of $\map{C}_{A_i \to B_{\pi''(i)} }$ and $\map{C}_{A_i \nrightarrow B_{\pi''(i)}}$, which gives a bound of the diamond-norm distance between the channels:
\begin{align}\label{eq:CCdchi}
    \| \map{C}_{A_i \to B_{\pi''(i) }} - \map{C}_{A_i \nrightarrow B_{\pi''(i)} } \|_\diamond \leq d_{A_i} \left\| C_{A_i,B_{\pi''(i)}} - \frac{I_{A_i}}{d_{A_i}}\otimes C_{B_{\pi''(i)}} \right\|_1 \leq 2d_{A_i} \chi_-
\end{align}

Combining Eqs. (\ref{eq:CDdia}) and (\ref{eq:CCdchi}), we obtain
\begin{align}
    \| \map{C} - \map{D} \|_\diamond \leq 2 \chi_- \sum_{i \notin A_{\rm matched}}d_{A_i} \leq 2 n d_A \varepsilon_0
\end{align}
where $d_A := \max_i d_{A_i}$. The query complexity of \autoref{alg:memoryless} is given by \autoref{lem:ind} as
\begin{align}
    N = O\left(d_A^4 d_B^4 \lambda_{\min}^{-2} \varepsilon_0^{-2} \log(d_A d_B \kappa_0^{-1}) \right)
\end{align}
Let $\kappa := n^2 \kappa_0$ and $\varepsilon := 2 n d_A \varepsilon_0$, we have
\begin{align}
    N = O\left(n^2 d_A^6 d_B^4 \lambda_{\min}^{-2} \varepsilon^{-2} \log(n d_A d_B \kappa^{-1}) \right)
\end{align}
and with probability $1-\kappa$,  $\| \map{C} - \map{D} \|_\diamond \leq \varepsilon$. The running time of \autoref{alg:memoryless} is in the same form as \autoref{alg:order_linear}.

\end{document}